\documentclass[12pt]{article}

\usepackage{amsmath,amsfonts,amssymb,amsthm,graphicx,cite}
\usepackage{mathrsfs}

\numberwithin{equation}{section}

\begin{document}

\newtheorem{theorem}{Theorem}[section]
\newtheorem{corollary}[theorem]{Corollary}
\newtheorem{lemma}[theorem]{Lemma}
\newtheorem{proposition}[theorem]{Proposition}

\newcommand{\adiffop}{A$\Delta$O}
\newcommand{\adiffops}{A$\Delta$Os}

\newcommand{\be}{\begin{equation}}
\newcommand{\ee}{\end{equation}}
\newcommand{\bea}{\begin{eqnarray}}
\newcommand{\eea}{\end{eqnarray}}
\newcommand{\sh}{{\rm sh}}
\newcommand{\ch}{{\rm ch}}
\newcommand{\einde}{$\ \ \ \Box$ \vspace{5mm}}
\newcommand{\De}{\Delta}
\newcommand{\de}{\delta}
\newcommand{\te}{\tilde{e}}
\newcommand{\ty}{\tilde{y}}
\newcommand{\Z}{{\mathbb Z}}
\newcommand{\N}{{\mathbb N}}
\newcommand{\C}{{\mathbb C}}
\newcommand{\Cs}{{\mathbb C}^{*}}
\newcommand{\R}{{\mathbb R}}
\newcommand{\Q}{{\mathbb Q}}
\newcommand{\T}{{\mathbb T}}
\newcommand{\re}{{\rm Re}\, }
\newcommand{\im}{{\rm Im}\, }
\newcommand{\cW}{{\cal W}}
\newcommand{\cJ}{{\cal J}}
\newcommand{\cE}{{\cal E}}
\newcommand{\cA}{{\cal A}}
\newcommand{\cR}{{\cal R}}
\newcommand{\cP}{{\cal P}}
\newcommand{\cM}{{\cal M}}
\newcommand{\cN}{{\cal N}}
\newcommand{\cI}{{\cal I}}
\newcommand{\cMs}{{\cal M}^{*}}
\newcommand{\cB}{{\cal B}}
\newcommand{\cD}{{\cal D}}
\newcommand{\cC}{{\cal C}}
\newcommand{\cL}{{\cal L}}
\newcommand{\cF}{{\cal F}}
\newcommand{\cG}{{\cal G}}
\newcommand{\cH}{{\cal H}}
\newcommand{\cO}{{\cal O}}
\newcommand{\cS}{{\cal S}}
\newcommand{\cT}{{\cal T}}
\newcommand{\cU}{{\cal U}}
\newcommand{\cQ}{{\cal Q}}
\newcommand{\cV}{{\cal V}}
\newcommand{\cK}{{\cal K}}
\newcommand{\cZ}{{\cal Z}}

\newcommand{\rE}{{\rm E}}
\newcommand{\rF}{{\rm F}}
\newcommand{\rI}{{\rm I}}
\newcommand{\rR}{{\rm R}}

\newcommand{\fm}{\mathfrak{m}}
\newcommand{\intR}{\int_{-\infty}^{\infty}}
\newcommand{\intI}{\int_{0}^{\pi/2r}}
\newcommand{\limp}{\lim_{\re x \to \infty}}
\newcommand{\limn}{\lim_{\re x \to -\infty}}
\newcommand{\limpn}{\lim_{|\re x| \to \infty}}
\newcommand{\diag}{{\rm diag}}
\newcommand{\Ln}{{\rm Ln}}
\newcommand{\Arg}{{\rm Arg}}
\newcommand{\LHP}{{\rm LHP}}
\newcommand{\RHP}{{\rm RHP}}
\newcommand{\UHP}{{\rm UHP}}
\newcommand{\Res}{{\rm Res}}
\newcommand{\ep}{\epsilon}
\newcommand{\ga}{\gamma}
\newcommand{\sing}{{\rm sing}}

\title{Joint eigenfunctions for the relativistic Calogero-Moser Hamiltonians of hyperbolic type. II. The two- and three-variable cases}
\author{Martin Halln\"as\footnote{E-mail: M.A.Hallnas@lboro.ac.uk} \\Department of Mathematical Sciences, \\ Loughborough University, Leicestershire LE11 3TU, UK \\ and \\Simon Ruijsenaars \\ School of Mathematics, \\ University of Leeds, Leeds LS2 9JT, UK}

\date{\today}

\maketitle

\begin{abstract}
In a previous paper we introduced and developed a recursive construction of joint eigenfunctions $J_N(a_+,a_-,b;x,y)$ for the Hamiltonians of the hyperbolic relativistic Calogero-Moser system with arbitrary particle number $N$.
In this paper we focus on the cases $N=2$ and $N=3$, and establish a number of conjectured features of the corresponding joint eigenfunctions. More specifically, choosing $a_+,a_-$ positive, we prove that $J_2(b;x,y)$ and $J_3(b;x,y)$ extend to globally meromorphic functions that satisfy various invariance properties as well as a duality relation. We also obtain detailed information on the asymptotic behavior of similarity transformed functions $\rE_2(b;x,y)$ and $\rE_3(b;x,y)$. In particular, we determine the dominant asymptotics for  $y_1-y_2\to\infty$ and $y_1-y_2,y_2-y_3\to\infty$, resp., from which the conjectured factorized scattering can be read off.
\end{abstract}

\tableofcontents

\section{Introduction}
In a previous paper \cite{HR14}, we initiated a recursive scheme for constructing joint eigenfunctions $J_N(a_+,a_-,b;x,y)$ of the commuting analytic difference operators (henceforth A$\De$Os) associated with the integrable $N$-particle quantum systems of hyperbolic relativistic Calogero-Moser type. As mentioned in the introduction of that paper, the possible existence of such a recursive scheme was suggested by earlier work on related integrable quantum systems, including the non-relativistic Calogero-Moser systems and the Toda systems of non-relativistic and relativistic type. (In~\cite{HR12} we detailed the connections between these systems and their associated kernel functions.) Accordingly, our starting point owes much to this pioneering work. It includes various papers by Gerasimov, Kharchev, Lebedev, Oblezin and Semenov-Tian-Shansky; the work of this group of authors can be traced back from what appears to be the most recent paper~\cite{GLO14}. (The first recursive construction for the Jack polynomials seems to occur in Section~5 of~\cite{S97}; the author informed us that it dates back to his 1989 PhD thesis. Recently, we also learned about a recursive construction of eigenfunctions for the rational Calogero-Moser system due to Guhr and Kohler \cite{GK02}.)

In our previous paper we established holomorphy domains and uniform decay bounds that were sufficient for proving that the scheme provides well-defined functions $J_N$ that satisfy the expected joint eigenvalue equations. We also presented an introduction to the general setting at issue, and information on the hyperbolic gamma function and related functions that enter in the recursive scheme. We shall make use of this information without further ado, referring back to sections and equations in \cite{HR14} by using a prefix~I.

As outlined in I Section 7, numerous aspects of the recursive scheme, associated with conjectured features of the joint eigenfunctions $J_N$, remain to be investigated. In the present paper, we deduce a rather comprehensive picture of the joint eigenfunctions in the $N=2$ and $N=3$ cases. Indeed, we establish global meromorphy, a number of invariance properties and a duality relation, and undertake a detailed study of their asymptotic behavior. For the $N=2$ case, nearly all of the results were already obtained in~\cite{R11}. The point of rederiving them here is not only to render them more accessible in the present context, but  also to switch from the flow chart of~\cite{R11} to methods and arguments that allow a generalization to $N>2$.

We proceed to sketch the main results and organization of this paper in more detail. With a view towards making it more self-contained, we briefly recall some key constructions and results from I as we go along.
Throughout the paper we take $a_+,a_-\in (0,\infty)$, use further parameters
\be\label{aconv}
\alpha\equiv 2\pi/a_+a_-,\ \ \ \ a\equiv (a_++a_-)/2,
\ee 
\be\label{asl}
a_s\equiv\min (a_+,a_-),\ \ \ a_l\equiv\max (a_+,a_-),
\ee
 and work with $b$-values in the strip
\be
S_a\equiv \{b\in\C \mid \re b\in (0,2a)\}.
\ee

Section \ref{Sec2} is devoted to the step from $N=1$ to $N=2$. From I Section 4, we recall that the first step $J_1\to J_2$ of the recursive scheme yields the representation
\be\label{J2}
J_2(b;x,y) = \exp(i\alpha y_2(x_1+x_2))\int_\R dz I_2(b;x,y,z),\ \ \ b\in S_a,\ x,y\in\R^2,
\ee
with integrand
\be\label{I2}
I_2(b;x,y,z)\equiv J_1(z,y_1-y_2)\cS^\sharp_2(b;x,z)=\exp(i\alpha z(y_1-y_2))\prod_{j=1}^2\frac{G(x_j-z-ib/2)}{G(x_j-z+ib/2)},
\ee
where $G(z)\equiv G(a_+,a_-;z)$ denotes the hyperbolic gamma function, reviewed in I Appendix A. (Here and below, we suppress the dependence on the parameters $a_+$, $a_-$, whenever this is not likely to cause ambiguities; the dependence on $b$ is often omitted as well.)

Taking $z\to z+(x_1+x_2)/2$ in the integral on the right-hand side of~\eqref{J2} and using the reflection equation I~(A.6) (viz., $G(-z)=1/G(z)$), we obtain another revealing representation, namely,
\be\label{J2cm}
\begin{split}
J_2(b;x,y)=&\exp(i\alpha (x_1+x_2)(y_1+y_2)/2)\\
&\times \int_\R dz \exp(i\alpha z(y_1-y_2))\prod_{\de_1,\de_2=+,-}G(\de_1z+\de_2(x_1-x_2)/2-ib/2).
\end{split}
\ee

Next, we note that the integrand $I_2$~\eqref{I2} can be written as a product of two factors, each of which involves only two hyperbolic gamma functions. Using the Plancherel relation and an explicit Fourier transform formula for factors of this type from \cite{R11}, we deduce a further representation for $J_2$ in Subsection \ref{Sec21}, which is related to the defining representation \eqref{J2} by the involution $(b,x,y)\mapsto (2a-b,y,x)$. As a consequence, we obtain a corresponding duality relation for $J_2(b;x,y)$, namely,
\be\label{J2drel}
J_2(b;x,y)=G(ia-ib)^2 J_2(2a-b;y,x).
\ee
Since $J_2$ has $S_2$-invariance in the variable~$x$ (as is clear from~\eqref{J2} and \eqref{I2}), this duality relation entails that $J_2$ is also $S_2$-symmetric in the variable~$y$ (which is at face value not clear from~\eqref{J2} and\eqref{I2}): 
\be\label{J2sym}
J_2(b;x,y)=J_2(b;\sigma x,\tau y),\ \ \ (\sigma,\tau)\in S_2\times S_2.
\ee
(Alternatively, the $y$-symmetry can be seen from~\eqref{J2cm}.)

Subsection~2.2 is devoted to global holomorphy and meromorphy features. We recall that a simple contour shift procedure reveals   that for $y\in\R^2$ the function $J_2(b;x,y)$ is holomorphic in $(b,x)$ on the domain
\be\label{D2}
D_2\equiv \{(b,x)\in S_a\times\C^2\mid |\im(x_1-x_2)|<2a-\re b\},
\ee
cf.~I~Proposition~4.1. Moreover, starting from the representation~\eqref{J2cm}, we concluded that $J_2$ has an analytic continuation to all $y\in\C^2$ satisfying $|\im(y_1-y_2)|<\re b$, thus arriving at the holomorphy domain
\be\label{cD2}
\cD_2\equiv \{(b,x,y)\in S_a\times\C^2\times\C^2\mid (b,x)\in D_2, |\im(y_1-y_2)|<\re b\}.
\ee

In Subsection \ref{Sec22} we improve these results by showing that the function $J_2(b;x,y)$ has a meromorphic extension to $S_a\times\C^2\times\C^2$, and we also determine the locations of its poles and bounds on their orders. 
To this end, we make use of the entire function $E(z)\equiv E(a_+,a_-;z)$, reviewed in Appendix~A. Specifically, introducing
\be\label{cP2}
\cP_2(b;x,y)\equiv J_2(b;x,y)\prod_{\de=+,-}E(\de(x_1-x_2)+ib-ia)E(\de(y_1-y_2)+ia-ib),
\ee
we show that the product function $\cP_2(b;x,y)$ has a holomorphic continuation to all $(b,x,y)\in S_a\times\C^2\times\C^2$. Since the zero locations and orders of~$E(z)$ are explicitly known, this yields the information on the polar divisor of $J_2$ just mentioned. 

Now in Appendix~B of~\cite{R99} a quite general result was obtained, from which these holomorphy results can also be derived. In fact, it has the stronger consequence that $\cP_2(b;x,y)$ is entire in~$b$ as well, and holomorphic for $a_+$ and~$a_-$ varying over the (open) right half plane. (The link to \cite{R99} can be gleaned from Section~3 in~\cite{R11}.) However, the methods used in~\cite{R99} give rise to insurmountable difficulties for the multi-variable case. 

By contrast, our present method of proof does extend to $N>2$. It involves some simple key ideas that are at risk of getting obscured by the inevitable technicalities required for their implementation. At this point it is therefore expedient to digress and isolate these ideas. (The reader may wish to skip to \eqref{c} at first reading and refer back to the following when the need arises.)

A key ingredient is Bochner's theorem on analytic completion of tube domains. (See Chapter~5 of the monograph~\cite{BM48} for a detailed account of Bochner's original proof in~\cite{B38}.)
For convenience we use the definition that a tube $\cT\subset\C^M$, $M\geq 1$, is any set of points $z=(z_1,\ldots,z_M)$, that can be represented in the form
\be\label{trep}
(\im z_1,\ldots,\im z_M)\in \cB,\ \ \ \re z_j\in\R,\ \ j=1,\ldots,M,
\ee
for some subset $\cB\subset\R^M$, called the base of $\cT$. In the mathematical literature it is customary to have the imaginary rather than the real parts of the complex variables vary over all of~$\R$, but this is clearly just a matter of convention; we actually need the latter convention for the dependence on the coupling parameter~$b$. We shall make use of Bochner's theorem in the following form.

\begin{theorem}[Bochner \cite{B38}]\label{BThm}
Every function that is holomorphic in a tube $\cT$ with an open, connected base~$\cB$ has a holomorphic continuation to the tube $\cT_c$  whose base~$\cB_c$ is the convex hull of~$\cB$. 
\end{theorem}

We proceed to sketch how we use this theorem to deduce holomorphy of~$\cP_2(b;x,y)$ in $S_a\times\C^2\times\C^2$, restricting attention to those steps in the reasoning that have   counterparts for $N>2$. This will enable us to shorten our account for the case $N=3$ in Subsection~3.2, and show what needs to be supplied for $N>3$.

First, we point out that the domain~$\cD_2$~\eqref{cD2} is a tube with respect to the variables $(ib,x,y)$, with an open, connected base
\be\label{cB2}
\cB_2\equiv \{(\re b, \im x,\im y)\in (0,2a)\times\R^2\times \R^2\mid |\im (x_1-x_2)|< 2a-\re b,|\im (y_1-y_2)|< \re b\}. 
\ee
Let us now assume that $\cP_2(b;x,y)$ has a holomorphic continuation to the tube with base
\be\label{cBep2}
\cB_2(\epsilon_2)\equiv \{(\re b, \im x,\im y)\in (0,\epsilon_2)\times\R^2\times \R^2\mid |\im (y_1-y_2)|< \re b\},\ \ \ \epsilon_2\in(0,a). 
\ee
Then it follows from the definition \eqref{cP2} of $\cP_2(b;x,y)$ and the duality relation~\eqref{J2drel}
that $\cP_2(b;x,y)$ also has a holomorphic continuation to the tube with base
\be\label{hcBep2}
\hat{\cB}_2(\epsilon_2)\equiv \{(\re b, \im x,\im y)\in (2a-\epsilon_2,2a)\times\R^2\times \R^2\mid |\im (x_1-x_2)|<2a- \re b\}. 
\ee
Indeed, the map $(b,x,y)\mapsto (2a-b,y,x)$ yields a bijection between $\cB_2(\epsilon_2)$  and~$\hat{\cB}_2(\epsilon_2)$, and both sets have a non-empty intersection with $\cB_2$.

We can now invoke Bochner's theorem applied to the tube with open, connected base
\be\label{cBunion}
\cB_2^u\equiv \cB_2\cup \cB_2(\epsilon_2)\cup \hat{\cB}_2(\epsilon_2).
\ee
This yields holomorphy of $\cP_2(b;x,y)$ in the tube whose base is the convex hull of the union~$\cB_2^u$.
It is not hard to see that the latter base is given by
\be\label{cB2h}
\cB_2^h
\equiv \{(\re b, \im x,\im y)\in (0,2a)\times\R^2\times \R^2\},
\ee
so that this  tube is the holomorphy domain~$S_a\times\C^2\times\C^2$ announced above.
 Specifically, for each $b\in S_a$, there clearly exist $\lambda\in(0,1)$, $b_-$ with $\re b_-\in(0,\epsilon_2)$, and~$b_+$ with $\re b_+\in(2a-\epsilon_2,2a)$ 
 such that
\be
b=\lambda b_-+(1-\lambda)b_+ .
\ee
As required, we can therefore write any $(b,x,y)\in S_a\times\C^2\times\C^2$ as a convex combination
\be\label{bcon}
(b,x,y)=\lambda\big(b_-,\lambda^{-1}x,0\big)+(1-\lambda)\big(b_+,0,(1-\lambda)^{-1}y\big).
\ee

It remains to prove our assumption (above~\eqref{cBep2}) that $\cP_2(b;x,y)$ is entire in~$x$ for $\re b$ sufficiently small and $|\im (y_1-y_2)|< \re b$. We do so by exploiting one of the A$\De$Es satisfied by $J_2(b;x,y)$, cf.~I Proposition 4.2. This involves a similarity transformation to the corresponding A$\De$E for $\cP_2(b;x,y)$, which leads to coefficients involving the (rational) gamma function, cf.~Lemma~\ref{Lemma:cP2eigeeq}. 

In Subsection~2.3 we collect results concerning the asymptotic behavior of a function $\rE_2(b;x,y)$ that is another similarity transform of~$J_2(b;x,y)$. To sketch these results, we first recall the generalized Harish-Chandra $c$-function
\be\label{c}
c(b;z)\equiv \frac{G(z+ia-ib)}{G(z+ia)},
\ee
and its multivariate version
\be\label{CN}
C_N(b;x)\equiv \prod_{1\leq j<k\leq N}c(b;x_j-x_k),\ \ \ N\geq 2.
\ee
Introducing the phase function
\be\label{phi}
\phi(b)\equiv \exp(i\alpha b(b-2a)/4),
\ee
  the pertinent $J_2$-cousin is given by
\be\label{rE2}
\rE_2(b;x,y)\equiv \frac{\phi(b)G(ib-ia)}{\sqrt{a_+a_-}}\frac{J_2(b;x,y)}{C_2(b;x)C_2(2a-b;y)}.
\ee
This function is particularly suitable for Hilbert space purposes.  We deduce  its dominant asymptotics for $y_1-y_2\to\infty$, namely,
\be\label{rE2as}
\rE_2(b;x,y)\sim \rE_2^{{\rm as}}(b;x,y)\equiv \exp(i\alpha(x_1y_1+x_2y_2))-u(b;x_2-x_1)\exp(i\alpha(x_2y_1+x_1y_2)),
\ee
where $u$ is the scattering function,
\be\label{u}
u(b;z)\equiv -\frac{c(b;z)}{c(b;-z)}=-\prod_{\de=+,-}\frac{G(z+\de i(a-b))}{G(z+\de ia)},
\ee
and we obtain a bound on the remainder, cf.~Proposition~\ref{Prop:rE2as}. In Proposition~\ref{Prop:rE2b} we also establish a uniform bound on $\rE_2(b;x,y)$ for $(x,y)\in\C^2\times\R^2$ satisfying $\im(x_1-x_2)\in(-a_s,0]$ and $y_1-y_2\geq 0$, which is needed to handle the $N=3$ case.

Section \ref{Sec3} is concerned with the step from $N = 2$ to $N = 3$. It is structured in parallel with Section~2, but several new ingredients and technical difficulties arise.
To begin with, we recall that to construct $J_3$ from~$J_2$ in I Section 5, we started from the integrand
\be\label{I3}
I_3(b;x,y,z)\equiv  \cS^\sharp_3(b;x,z)W_2(b;z)J_2(b;z,(y_1-y_3,y_2-y_3)),
\ee
with weight function
\be\label{W2}
W_2(b;z)\equiv 1/C_2(b;z)C_2(b;-z),
\ee
and kernel function
\be\label{cS3}
\cS^\sharp_3(b;x,z)\equiv\prod_{j=1}^3\prod_{k=1}^2\frac{G(x_j-z_k-ib/2)}{G(x_j-z_k+ib/2)}.
\ee
More precisely, from I (5.6) we have the representation
\be\label{J3}
J_3(b;x,y) =  \exp(i\alpha y_3(x_1+x_2+x_3))\int_{G_2} dz\, I_3(b;x,y,z),\ \ \ b\in S_a,\ x,y\in\R^3,
\ee
where we have introduced the `Weyl chamber',
\be
G_2\equiv \{z\in\R^2\mid z_2<z_1\}.
\ee

To derive the counterpart of~\eqref{J2cm} (and for later purposes), we define
\be\label{XY}
X_3\equiv \frac13\sum_{j=1}^3x_j,\ \ Y_3\equiv \frac13\sum_{j=1}^3y_j,\ \  \tilde{x}_j\equiv  x_j-X_3,\ \ \ \tilde{y}_j\equiv  y_j-Y_3,\ \ \ j=1,2,3.
\ee
Taking $z\to z+X_3$ in the integral on the right-hand side of~\eqref{J3}, and then using~\eqref{J2cm}, we obtain
\be\label{J3cm}
\begin{split}
J_3(b;x,y)=&\exp(3i\alpha X_3Y_3)\\
& \times \int_{G_2} dz\, \cS^\sharp_3(b;\tilde{x},z)W_2(b;z)J_2(b;z,(y_1-y_3,y_2-y_3)).
\end{split}
\ee
Note that the integral yields a function that depends only on the differences $x_j-x_k$ and $y_j-y_k$, $j,k=1,2,3$.

As a principal result of Subsection \ref{Sec31}, we deduce a novel representation for $J_3$, related to \eqref{J3} by taking $(b,x,y)\mapsto (2a-b,y,x)$. To generalize our approach in the $N=2$  case, we rely on results from our recent joint paper \cite{HR15} on product formulas for conical functions. Specifically, starting from the Plancherel relation for a generalized Fourier transform, we make use of the remarkable fact that $J_2(b;z,y)$ is an eigenfunction of the integral operator whose kernel is the product of the function
\be\label{cS2}
\cS_2(b;x,z)\equiv \prod_{j,k=1}^2\frac{G(x_j-z_k-ib/2)}{G(x_j-z_k+ib/2)},
\ee
and the weight function $W_2(b;z)$, with the eigenvalue given explicitly by a product of $y$-dependent $G$-functions. (This can be viewed as the $N=2$ counterpart of the Fourier transform formula used for $N=1$.) We also need to invoke the closely related explicit generalized eigenfunction expansion for the integral operator on $L^2(G_2,dx)$ with kernel $W_2(b;x)^{1/2}\cS_2(b;x,y)W_2(b;y)^{1/2}$ from~\cite{HR15}. 

Once the new representation for $J_3$ has been established, the $N=3$ counterparts of~\eqref{J2drel} and~\eqref{J2sym} readily follow. Specifically, they read
\be\label{J3drel}
J_3(b;x,y)=G(ia-ib)^6 J_3(2a-b;y,x),
\ee
 and
\be\label{J3sym}
J_3(b;x,y)=J_3(b;\sigma x,\tau y),\ \ \ (\sigma,\tau)\in S_3\times S_3.
\ee
(Note that in this case the $y$-symmetry is not at all clear from the `center-of-mass' representation~\eqref{J3cm}.)

Turning to Subsection~3.2, we recall that in I Proposition 5.1 we proved, by shifting the two contours in \eqref{J3} simultaneously, that $J_3(b;x,y)$ (with $y\in\R^3$ fixed) is holomorphic in
\be\label{D3}
D_3\equiv\Big\{(b,x)\in S_a\times\C^3\mid \max_{1\leq j<k\leq 3}|\im(x_j-x_k)|<2a-\re b\Big\}.
\ee
To conclude analytic continuation to $y\in\C^3$ such that $|\im(y_j-y_k)|<\re b$, $1\leq j<k\leq 3$, we arrived at a subdomain of $D_3$ for the dependence on $(b,x)$. Specifically, using   the notation~\eqref{XY}, 
we needed the restricted domain
\be\label{D3r}
D_3^r\equiv \{ (b,x)\in S_a\times \C^3 \mid |\im \tilde{x}_j|<a-\re b/2,\ \ j=1,2,3 \}\subset D_3.
\ee
In I Proposition 5.4 we then showed that $J_3(b;x,y)$ is holomorphic in the domain 
\be\label{cD3}
\cD_3\equiv \Big\{ (b,x,y)\in D_3^r\times \C^3 \mid  \max_{1\leq j<k\leq 3}|\im (y_j-y_k)|<\re b \Big\}.
\ee

With these preliminaries in place, we can follow the $N=2$ flow chart. Defining the counterpart
\be\label{cP3}
\cP_3(b;x,y)\equiv J_3(b;x,y)\prod_{1\leq j<k\leq 3}\, \prod_{\de=+,-}E(\de(x_j-x_k)+ib-ia)E(\de(y_j-y_k)+ia-ib),
\ee
of~\eqref{cP2}, this leads to the conclusion that the functions $\cP_3(b;x,y)$/$J_3(b;x,y)$ extend from~$\cD_3$ to holomorphic/meromorphic functions on all of $S_a\times\C^3\times\C^3$, yielding as a corollary the locations of the $J_3$-poles and bounds on their orders. More specifically, there are natural $N=3$ analogs of the bases~\eqref{cB2}--\eqref{cB2h}, and the role of the $J_2$-duality relation~\eqref{J2drel} in the $N=2$ reasoning is played by~\eqref{J3drel}. 

In order to prove the critical assumption that $\cP_3(b;x,y)$ is entire in~$x$ for $\re b$ sufficiently small, however, it is necessary to supplement the consideration of the pertinent $\cP_3$-A$\De$E by a further inductive reasoning, exploiting once more Bochner's Theorem~1.1. (We intend to generalize this part of the argument to arbitrary~$N$ in the next paper of this series.)

In Subsection \ref{Sec33} we consider the asymptotic behavior of the function
\be\label{rE3}
\rE_3(b;x,y)\equiv \left(\frac{\phi(b)G(ib-ia)}{\sqrt{a_+a_-}}\right)^3\frac{J_3(b;x,y)}{C_3(b;x)C_3(2a-b;y)}.
\ee
This involves considerable technicalities, with an important auxiliary result relegated to Lemma~\ref{Lem:aux3}. A highlight is that Theorem~\ref{Thm:rE3as} implies an explicit formula for the dominant asymptotics as $y_1-y_2,y_2-y_3\to\infty$, viz.,
\be\label{E3sc}
\rE_3(b;x,y)\sim \rE_3^{{\rm as}}(b;x,y)\equiv \sum_{\sigma\in S_3}\prod_{\substack{j<k\\\sigma^{-1}(j)>\sigma^{-1}(k)}}(-u(b;x_k-x_j))\cdot\exp\Big(i\alpha \sum_{j=1}^3 x_{\sigma(j)}y_j\Big).
\ee
Indeed, this formula amounts to the factorized scattering conjectured in~I~(7.6).
 With a view towards generalizing our results concerning asymptotics to $N>3$, we also derive a uniform bound on $\rE_3(b;x,y)$ for suitably restricted $(x,y)\in\C^3\times\R^3$, cf.~Theorem~\ref{Thm:ubound}.

\section{The step from $N=1$ to $N=2$}\label{Sec2}

\subsection{Invariance properties and a duality relation}\label{Sec21}
We begin this subsection by collecting some invariance properties for $J_2$, which we have occasion to invoke below.

\begin{proposition}\label{Prop:J2sym}
For all $(b,x,y)\in\cD_2$~\eqref{cD2} and $\eta\in\C$, we have
\be\label{J2ri}
J_2(b;x,y)=J_2(b;-x,-y),
\ee 
\be\label{J2hom}
\begin{split}
J_2(b;x,y)=& \exp(-i\alpha\eta(y_1+y_2))J_2(b;(x_1+\eta,x_2+\eta),y)\\
& =\exp(-i\alpha\eta(x_1+x_2))J_2(b;x,(y_1+\eta,y_2+\eta)) .
\end{split}
\ee
\end{proposition}
\begin{proof}
To begin with, we assume $x,y\in\R^2$. 
It is clear from  the reflection equation I (A.6) for $G(z)$ (namely, $G(-z)=1/G(z)$) that the integrand $I_2$~\eqref{I2} satisfies 
\be\label{I2ri}
I_2(-x,-y,-z)=I_2(x,y,z).
\ee
Taking $z\to -z$ in the defining representation \eqref{J2}, the invariance property \eqref{J2ri} is immediate from \eqref{I2ri}. Assuming also $\eta\in\R$, we obtain~\eqref{J2hom} from the alternative representation~\eqref{J2cm}.  Clearly, \eqref{J2ri}--\eqref{J2hom} are preserved under analytic continuation, and so the proposition follows.
\end{proof}

We proceed to deduce a new representation for $J_2$, which is related to  \eqref{J2} by the involution $(b,x,y)\mapsto (2a-b,y,x)$. We start from the Plancherel relation
\be\label{Plancherel}
\int_\R dz f(z)g(z)=\int_\R dp \hat{f}(p)\hat{g}(-p),\ \ \ f,g\in L^2(\R)\cap L^1(\R),
\ee
with the Fourier transform defined by
\be\label{Ftrans}
\hat{h}(p)=\left(\frac{\alpha}{2\pi}\right)^{1/2}\int_\R dz\exp(i\alpha pz)h(z),\ \ \ h=f,g.
\ee
Choosing
\be\label{fg}
f(z)=\frac{G(x_1-z-ib/2)}{G(x_1-z+ib/2)},\  \ g(z)=\exp\big(i\alpha (y_2(x_1+x_2)+ z(y_1-y_2))\big)\frac{G(x_2-z-ib/2)}{G(x_2-z+ib/2)},
\ee
the left-hand side of \eqref{Plancherel} coincides with  the $J_2$-representation \eqref{J2}. We can calculate the Fourier transforms of these two functions by  using  the Fourier transform formula \eqref{Fform}. Indeed, setting $\mu=x_1-ib/2$ and $\nu=x_1+ib/2$, and invoking the reflection equation I (A.6), we obtain
\be\label{ftr}
\hat{f}(p)=G(ia-ib)\exp(i\alpha x_1 p)\prod_{\de=+,-}G(\de p-ia+ib/2).
\ee
Swapping~$x_1$ and~$x_2$, and taking $p\to p+y_1-y_2$, we deduce
\be
\hat{g}(-p)=G(ia-ib)\exp(i\alpha (x_2y_1-x_1y_2-x_2 p))\prod_{\de=+,-}G(\de(y_1-y_2-p)-ia+ib/2).
\ee
Substituting these expressions in the right-hand side of \eqref{Plancherel} and taking $p\to p-y_2$, we get  the new representation
\be\label{J2d}
J_2(b;x,y)=G(ia-ib)^2 \exp(i\alpha x_2(y_1+y_2))\int_\R dp I_2(2a-b;y,x,p).
\ee
We are now prepared for  the following result.

\begin{proposition}\label{Prop:J2prop}
Letting $b\in S_a$ and $x,y\in\R^2$, the duality relation~\eqref{J2drel} and symmetry relation~\eqref{J2sym} hold true. 
\end{proposition}
\begin{proof}
Comparing \eqref{J2d} to the defining representation~\eqref{J2}, we obtain~\eqref{J2drel}. Now $S_2$-symmetry in~$x$ is immediate from~\eqref{J2}, so $S_2$-symmetry in~$y$ then follows from the dual representation \eqref{J2d} or directly from~\eqref{J2cm}.
\end{proof}

For completeness, we add that $J_2$ has a further duality property, namely,
\be\label{J2sd}
J_2(b;y,x)=J_2(b;x,y)\prod_{\de=+,-}G(\de(x_1-x_2)-ia+ib)G(\de(y_1-y_2)+ia-ib).
\ee
It can be derived from~\eqref{Fform} in the same way as before, by starting from~\eqref{fg} with the denominators swapped. Indeed, this yields yet another $J_2$-representation. Taking $p\to p+(y_1-y_2)/2$  in the latter, it becomes
\be
\begin{split}
J_2(b;x,y)=& \exp(i\alpha (x_1+x_2)(y_1+y_2)/2)\prod_{\de=+,-}G(\de (x_2-x_2)+ia-ib)
\\
& \times \int_{\R} dp\, \prod_{\de=+,-}\frac{G(p+\de (x_d-y_d)/2-ia+ib/2)}{G(p+\de (x_d+y_d)/2+ia-ib/2)},
\end{split}
\ee
with $x_d\equiv x_1-x_2$ and $y_d\equiv y_1-y_2$. (The function defined by the integral is manifestly invariant under swapping~$x_d$ and~$y_d$\,; it is a multiple of the relativistic conical function $\cR(2a-b;x_d,y_d)$, cf.~Eq.~(1.3) in~\cite{R11}.) Formula~\eqref{J2sd} easily follows from this representation.

The additional duality feature~\eqref{J2sd} entails that the function~$E_2(b;x,y)$ given by~\eqref{rE2} is invariant under $x\leftrightarrow y$. We believe that this self-duality feature also holds for the $N=3$ counterpart $E_3(b;x,y)$~\eqref{rE3}, but so far a proof of this conjecture has not materialized.

\subsection{Global meromorphy}\label{Sec22}
In this subsection we show that the product function~$\cP_2(b;x,y)$~\eqref{cP2} has a holomorphic continuation from the domain~$\cD_2$~\eqref{cD2} to~$S_a\times\C^2\times\C^2$. To do so, we follow the flow chart outlined below~\eqref{cP2}.

We begin by noting that as a corollary of Propositions~\ref{Prop:J2sym} and~\ref{Prop:J2prop} we obtain
\be\label{cP2ref}
\cP_2(b;x,y)=\cP_2(b;-x,-y), \ \ \ \mathrm{(reflection~invariance)},
\ee
\be\label{cP2d}
\cP_2(b;x,y)=G(ia-ib)^2\cP_2(2a-b;y,x),\ \ \ \mathrm{(duality)},
\ee
\be\label{cP2pi}
\cP_2(b;x,y)=\cP_2(b;\sigma x,\tau y),\ \ \ (\sigma,\tau)\in S_2\times S_2,\ \ \ \mathrm{(permutation~invariance)}.
\ee
Indeed, the $E$-function product in \eqref{cP2} is invariant under the reflections $z\mapsto-z$, $z=x,y$, the map $(b,x,y)\mapsto(2a-b,y,x)$, as well as each of the four permutations $(x,y)\mapsto(\sigma x,\tau y)$, $(\sigma,\tau)\in S_2\times S_2$.

From the second $J_2$-duality feature~\eqref{J2sd} it also follows that we have
\be\label{cP2sd}
\cP_2(b;x,y)= \cP_2(b;y,x),\ \ \  {\rm (self-duality)}.
\ee
However, we shall avoid the use of this property, since we are so far unable to prove the expected self-duality for $\cP_3(b;x,y)$.

Next, as announced below~\eqref{bcon}, we are going to replace one of the eigenvalue equations for~$J_2$ in I~Proposition~4.2  by the corresponding eigenvalue equation for $\cP_2$. Specifically, we focus on the A$\De$E obtained by setting $k=1$ and choosing $\de\in\{+,-\}$ such that $a_{-\de}=a_s$ (recall~\eqref{asl}). Using henceforth the notation
\be
e_l(z)\equiv\exp(\pi z/a_l),\ \ \ s_l(z)\equiv\sinh(\pi z/a_l),
\ee
this equation reads
 \begin{multline}\label{J2eigeq}
V_2(b;x)J_2(b;x+ia_se_1,y)+V_2(b;\sigma_{12}x)J_2(b;x+ia_se_2,y)
\\
=\big(e_l(-2y_1)+e_l(-2y_2)\big)J_2(b;x,y).
\end{multline}
Here, we have $e_1\equiv (1,0),e_2\equiv (0,1)$, the map $\sigma_{12}$ swaps $x_1$ and~$x_2$, and the coefficient function is given by
\be
V_2(b;x)\equiv \frac{s_l(x_2-x_1-ib)}{s_l(x_2-x_1)}.
\ee
(To be quite precise, we have taken $(x,y)\to (-x,-y)$ in I~(4.10) with $k=1$ and used the reflection invariance~\eqref{J2ri}; cf.~also I~(1.21) and I~(1.9).)

We need to ensure that the $x_j$-shifts do not move the $J_2$-argument out of $\cD_2$~\eqref{cD2}. To this end and also for later purposes (in particular, to complete the definition of the base~$\cB_2(\epsilon_2)$~\eqref{cBep2}), we introduce the number
\be\label{ep2}
\epsilon_2\equiv a_l/2,
\ee
 the strip $S(\epsilon_2)$, where
\be\label{Sep}
S(\epsilon)\equiv \{ b\in S_a\mid \re b <\epsilon\}, \ \ \ \epsilon\in (0,a),
\ee
and the domains
\be\label{cA2}
\cA_2\equiv \{ x\in\C^2\mid v_1-v_2>-\re b\},
\ee
\be\label{cA2n}
\cA_2^{(n)}\equiv
\left\{
\begin{array}{ll}
  \{ x\in\C^2\mid |v_1-v_2|<a_s+\re b\} ,  & n=1   , \\
   \{ x\in\cA_2\mid v_1-v_2<na_s+\re b\},  & n=0,2,3,\ldots  .
\end{array}
\right.
\ee
Here and from now on, we use the notation
\be
v\equiv \im x,\ \ \ x\in\C^M.
\ee
Next, we introduce
\be\label{D2p}
D_2^{(+)}\equiv \big\{ (b,x)\in S(\epsilon_2)\times\C^2\mid x\in\cA_2\big\}, \ee
\be\label{D2n}
D_2^{(n)}\equiv \big\{ (b,x)\in S(\epsilon_2)\times\C^2\mid x\in\cA_2^{(n)}\big\},  
\ee
\be\label{cD2p}
\cD_2^{(+)}\equiv \big\{ (b,x,y)\in D_2^{(+)}\times\C^2\mid |\im (y_1-y_2)|<\re b\big\},
\ee
\be\label{cD2n}
\cD_2^{(n)}\equiv \big\{ (b,x,y)\in D_2^{(n)}\times\C^2\mid |\im (y_1-y_2)|<\re b\big\},
\ee
and note that we have inclusions
\be\label{inclu}
D_2^{(1)}\subset D_2,\ \ \ \cD_2^{(1)}\subset \cD_2.
\ee
(Indeed, since $b$ belongs to $S(\epsilon_2)$, we have $a_s+\re b< a_s+a_l-\re b$.)

We are now prepared for the following lemma.

\begin{lemma}\label{Lemma:cP2eigeeq}
Letting $(b,x,y)\in\cD_2^{(0)}$, we have the eigenvalue equation
 \begin{multline}\label{cP2eigeq}
\cV_2(b;x)\cP_2(b;x+ia_se_1,y)+\cV_2(b;\sigma_{12}x)\cP_2(b;x+ia_se_2,y)
\\
=\big(e_l(-2y_1)+e_l(-2y_2)\big)\cP_2(b;x,y),
\end{multline}
where the coefficient function is given by
\be\label{cV}
\begin{split}
\cV_2(b;x) &\equiv- i\pi\frac{\exp(i(2x_2-2x_1-ia_s)K_l)}{s_l(x_2-x_1)}\\
&\quad\times\left[\Gamma\left(\frac{i}{a_l}(x_2-x_1-ib)\right)\Gamma\left(\frac{i}{a_l}(x_2-x_1+ib-2ia)\right)\right]^{-1},
\end{split}
\ee
with
\be\label{slKl}
 K_l\equiv \frac{1}{2a_l}\ln\left(\frac{a_s}{a_l}\right).
\ee
\end{lemma}

\begin{proof}
Note first that for $(b,x,y)\in\cD_2^{(0)}$ the three arguments of~$J_2$ occurring in \eqref{J2eigeq}  belong to $\cD_2^{(1)}$, and thus to the holomorphy domain~$\cD_2$, cf.~\eqref{inclu}.
Next, using the pertinent A$\De$E \eqref{EADE} satisfied by $E(z)$ and the reflection equation for $\Gamma(z)$, we compute
\be\label{EpADE}
\begin{split}
\prod_{\de=+,-}\frac{E(\de t+ib-ia)}{E(\de t+ib-ia+\de ia_s)}&=i\pi\frac{\exp(i(-2t-ia_s)K_l)}{s_l(t+ib)}\\
&\quad\times\left[\Gamma\left(\frac{i}{a_l}(-t-ib)\right)\Gamma\left(\frac{i}{a_l}(-t+ib-2ia)\right)\right]^{-1}.
\end{split}
\ee
Using this, the A$\De$E \eqref{cP2eigeq} readily follows from~\eqref{J2eigeq}.
\end{proof}

Now we are ready for the proof of the main result of this subsection.

\begin{proposition}\label{Prop:cP2ext}
The product function $\cP_2(b;x,y)$~\eqref{cP2} admits a holomorphic continuation from~$\cD_2$~\eqref{cD2} to $S_a\times\C^2\times\C^2$.
\end{proposition}
\begin{proof}
We begin by proving holomorphic continuation to~$\cD_2^{(+)}$~\eqref{cD2p}. To this end, 
we assume inductively that $\cP_2(b;x,y)$ is holomorphic in~$\cD_2^{(n)}$ with $n\geq 1$. (For $n=1$ the validity of the assumption follows from the inclusion~\eqref{inclu}.)
To establish holomorphic continuation to~$\cD_2^{(n+1)}$, we rewrite the eigenvalue equation~\eqref{cP2eigeq} in a more convenient form. Letting
\be
 \hat{\cV}_2(b;x)\equiv s_l(x_2-x_1)\cV_2(b;x)  ,
\ee
multiplying \eqref{cP2eigeq} by $s_l(x_2-x_1)$, and rearranging, we obtain 
 \be \label{cP2id}
 \begin{split}
\hat{\cV}_2(b;x)\cP_2(b;x+ia_se_1,y)= &\hat{\cV}_2(b;\sigma_{12}x)\cP_2(b;x+ia_se_2,y)
\\
&+s_l(x_2-x_1)\big(e_l(-2y_1)+e_l(-2y_2)\big)\cP_2(b;x,y).
\end{split}
\ee
Now $1/\Gamma(z)$ is an entire function with zeros at $z=-k,k\in\N$, so the function $\hat{\cV}_2(b;x)$ is entire as well, with zeros located at
\be\label{gamzeros}
x_1-x_2=-ib-ika_l,\ \ x_1-x_2=-2ia+ib-ika_l,\ \ \ k\in\N.
\ee
This implies, in particular, that $\hat{\cV}_2(b;x)$ is nonzero on $D^{(+)}_2$~\eqref{D2p}. 

We now assert that it is  enough to prove that the function $R_2(b;x,y)$ given by the right-hand side of \eqref{cP2id} is holomorphic for all points $(b,x,y)\in\cD_2^{(n)}$ satisfying
\be\label{vres2}
v_1-v_2\in ((n-1)a_s-\re b,na_s+\re b).
\ee
Indeed, this restriction yields a subdomain
\be
\cD_{2,r}^{(n)}\subset\cD_2^{(n)},
\ee
 whose $x$-translation over $ia_se_1$ equals~$\cD_{2,r}^{(n+1)}$, and~$\cD_{2,r}^{(n+1)}$  meets~$\cD_2^{(n)}$  for all points with~$
 v_1-v_2\in (na_s-\re b,na_s+\re b)$. Thus we obtain a holomorphic continuation to all of~$\cD_2^{(n+1)}$, as announced. 

To verify that $R_2(b;x,y)$ is indeed holomorphic in~$\cD_{2,r}^{(n)}$, we need only note that for $n=1$ both terms $\cP_2(b;x,y)$ and $\cP_2(b;x+ia_se_2,y)$ in $R_2(b;x,y)$ are holomorphic  in~$\cD_{2,r}^{(1)}$ by virtue of \eqref{inclu}, while for $n>1$ they are holomorphic in~$\cD_{2,r}^{(n)}$ thanks to the induction assumption. This completes the induction argument, so it follows that $\cP_2(b;x,y)$ has a holomorphic continuation to~$\cD_2^{(+)}$. 

Finally, we invoke the reflection invariance~\eqref{cP2ref} to deduce holomorphic continuation to the tube with base $\cB_2(\epsilon_2)$~\eqref{cBep2}. We can then follow the reasoning detailed below~\eqref{cBep2} to complete the proof of the proposition.
\end{proof}

\subsection{Asymptotics}\label{Sec23}
In this subsection we undertake a detailed study of the asymptotic behavior of the function $\rE_2(b;x,y)$ \eqref{rE2}.
To begin with, we note that the phase function~\eqref{phi} and scattering function~\eqref{u} satisfy
\be\label{phuinv}
\phi(2a-b)=\phi(b),\ \ \ u(2a-b;z)=u(b;z),
\ee
whereas  the $c$-function~\eqref{c} and its multivariate version $C_N$~\eqref{CN} are not invariant under this $b$-involution.
 Next, we invoke the $G$-function asymptotics I (A.14)--(A.16) to deduce the asymptotics of the $c$-function, namely, 
\be\label{cas}
|\phi(b)^{\mp 1}\exp(\pm\alpha bz/2)c(b;z)-1| \le C_1(\rho,b,\im z)\exp(-\alpha\rho|\re z|),\ \ \ \re z\to\pm\infty.
\ee
Here the decay rate $\rho$ can be chosen in $[a_s/2,a_s)$, and $C_1$ is continuous on $[a_s/2,a_s)\times S_a\times \R$. 
 It follows that the $u$-function satisfies 
\be\label{uas}
|u(b;z)\phi(b)^{\mp 2}+1|
\le C_2(\rho,b,\im z)\exp(-\alpha\rho|\re z|),\ \ \ \re z\to\pm\infty,
\ee
 with $C_2$  continuous on $[a_s/2,a_s)\times S_a\times \R$.
Moreover, from \eqref{u} it is clear that
\be\label{urefl}
u(b;z)u(b;-z)=1,
\ee
and, by the reflection equation I (A.6) and the conjugation relation I (A.9), we have
\be\label{umod}
|u(b;z)|=1,\ \ \ b,z\in\R.
\ee

 From \eqref{c}--\eqref{CN}, Proposition \ref{Prop:cP2ext}, and \eqref{GE}--\eqref{pkl}, we deduce that $\rE_2(b;x,y)$ is meromorphic in $x$ and $y$, with $b$-independent pole locations
\be\label{rE2ps1}
z_1-z_2=-2ia-ip_{kl},\ \ \ z=x,y,\ \ \ k,l\in\N,
\ee
and $b$-dependent poles located at
\be\label{rE2ps2}
z_1-z_2=ib+ip_{kl},\ \ z_1-z_2=2ia-ib+ip_{kl},\ \ \ z=x,y,\ \ \ k,l\in\N.
\ee
We collect further useful properties of $\rE_2$ in the following lemma. 

\begin{lemma}\label{Lemma:rE2prop}
For all $(b,x,y)\in S_a\times\C^2\times\C^2$ and $\eta\in\C$, the function $\rE_2(b;x,y)$~\eqref{rE2} satisfies
\be\label{rE2ri}
\rE_2(b;-x,-y)=u(b;x_1-x_2)u(b;y_1-y_2)\rE_2(b;x,y),
\ee
\be\label{rE2hom}
\begin{split}
\rE_2(b;x,y)=& \exp(-i\alpha\eta(y_1+y_2))\rE_2(b;(x_1+\eta,x_2+\eta),y)\\
& =\exp(-i\alpha\eta(x_1+x_2))\rE_2(b;x,(y_1+\eta,y_2+\eta)),
\end{split}
\ee
\be\label{rE2d}
\rE_2(b;x,y)=\rE_2(2a-b;y,x),
\ee
\be\label{rE2p}
\rE_2(b;\sigma x,\tau y)=(-u(b;x_1-x_2))^{|\sigma|}(-u(b;y_1-y_2))^{|\tau|}\rE_2(b;x,y),\ \ \ (\sigma,\tau)\in S_2\times S_2,
\ee
where $|\sigma|=0$  for $\sigma={\rm id}$ and $|\sigma|=1$ for~$\sigma=\sigma_{12}$. 
\end{lemma}
\begin{proof}
By global meromorphy, we need only check these features for $(b,x,y)\in (0,2a)\times\R^2\times\R^2$. The first two then readily follow from Proposition~\ref{Prop:J2sym}, using also \eqref{c}, \eqref{CN}, \eqref{u} and \eqref{phuinv}. Likewise, the last two follow from Proposition~\ref{Prop:J2prop}.
\end{proof}

In fact, as mentioned at the end of Subsection~2.1, we also have
\be\label{E2sd}
\rE_2(b;x,y)=\rE_2(b;y,x),
\ee
but we shall not invoke this self-duality feature.

Thanks to these symmetry properties, we need only establish the $y_1-y_2\to\infty$ asymptotics of $\rE_2$ to obtain a detailed picture of its asymptotic behavior. Indeed, from \eqref{rE2p} and the $u$-asymptotics \eqref{uas} the $y_1-y_2\to-\infty$ asymptotics easily follows, and the $x_1-x_2\to\pm\infty$ asymptotics can then be found via \eqref{rE2d}. 
 
Recalling from I (2.11) the kernel function
\be\label{cK2}
\cK_2^\sharp(b;x,z)\equiv C_2(b;x)^{-1}\cS_2^\sharp(b;x,z),
\ee
it is readily seen that \eqref{J2}--\eqref{I2} and \eqref{rE2} yield the representation
\be\label{rE2rep}
\rE_2(b;x,y)=\frac{\phi(b)G(ib-ia)}{\sqrt{a_+a_-}}\frac{\exp(i\alpha y_2(x_1+x_2))}{C_2(2a-b;y)}\int_\R dz \rI_2(b;x,y,z),\ \ \ b\in S_a,\ x,y\in\R^2,
\ee
with integrand
\be\label{rI2}
\rI_2(b;x,y,z)\equiv\exp(i\alpha z(y_1-y_2))\cK_2^\sharp(b;x,z).
\ee

Assuming $x_1\ne x_2$ until further notice, we now shift the contour $\R$ up by $a-\re b/2+r$, $r\in(0,a_s)$, so that we only  meet the simple poles at
\be\label{rI2ps}
z=x_m+ia-ib/2,\ \ \ m=1,2.
\ee
(The bound I (4.5) ensures that the shift causes no problems at the contour tails.)
Introducing the multiplier
\be\label{M2def}
M_2(b;y)\equiv \frac{\phi(b)}{c(2a-b;y_1-y_2)}\rho_2(b;y),
\ee
with
\be
\rho_2(b;y)\equiv \exp(-\alpha(a-b/2)(y_1-y_2)),
\ee
and the contour
\be\label{Cb}
C_b\equiv \R+i(a-\re b/2),
\ee
we are prepared for the following lemma.

\begin{lemma}\label{Lem:aux2}
Letting $(r,b)\in(0,a_s)\times S_a$ and $x,y\in\R^2$ with $x_1\ne x_2$, we have
\begin{multline}\label{rE2rep2}
\frac{\rE_2(b;x,y)}{M_2(b;y)}\exp(-i\alpha y_2(x_1+x_2))
=\frac{1}{\rho_2(b;y)}\frac{G(ib-ia)}{\sqrt{a_+a_-}}\int_{C_b+ir} dz\,\rI_2(b;x,y,z)
\\
+\exp(i\alpha x_1(y_1-y_2))-u(b;x_2-x_1)\exp(i\alpha x_2(y_1-y_2)).
\end{multline}
\end{lemma}
\begin{proof}
As just detailed, we shift contours in~\eqref{rE2rep}.
 Using the formula I~(A.13) for the residue of $G(z)$ at its simple pole $z=-ia$, we obtain\begin{multline}
2\pi i\, \Res\ \rI_2(x,y,z)\arrowvert_{z=x_m+ia-ib/2}\\
=\rho_2(b;y)\frac{\sqrt{a_+a_-}}{G(ib-ia)}\prod_{j<m}(-u(x_m-x_j))\cdot\exp(i\alpha x_m(y_1-y_2)).
\end{multline}
From this we easily get~\eqref{rE2rep2}.
\end{proof}

Even though we derived the representation~\eqref{rE2rep2} for $x_1\ne x_2$,
it is clearly valid for~$x_1=x_2$, too.
In point of fact, both $\rE_2(b;x,y)$ and $\rE_2^{{\rm as}}(b;x,y)$ (given by~\eqref{rE2as}) vanish for $x_1=x_2$. Indeed, recalling~\eqref{c} and~\eqref{u}, together with the simple zero/pole of $G(z)$ for $z=ia$/$z=-ia$, we obtain  
\be
1/c(b;0)=0,\ \ \ u(b;0)=1,\ \ \ b\in S_a,
\ee
from which this zero feature is plain. 
 
For $z$ on the contour $C_b+ir$, the integrand $\rI_2$~\eqref{rI2} decays exponentially with rate $\alpha(a-\re b/2+r)$ as $y_1-y_2\to\infty$. Moreover, from~\eqref{phuinv} and~\eqref{cas}   we get
\be\label{M2as}
M_2(b;y)=1+O(\exp(-\alpha \rho (y_1-y_2))),\ \ \ \rho\in [a_s/2,a_s),\ \ \  y_2-y_2\to\infty.
\ee
Combining these two observations with the representation \eqref{rE2rep2}, we are led to expect that the dominant asymptotics of $\rE_2$ for $y_1-y_2\to\infty$ is given by the function~$\rE_2^{{\rm as}}$ defined in~\eqref{rE2as}. This expectation is borne out and improved by the following proposition.

\begin{proposition}\label{Prop:rE2as}
Letting $(r,b)\in[a_s/2,a_s)\times S_a$, we have
\be\label{EEas}
\left|\left(\rE_2-\rE_2^{{\rm as}}\right)(b;x,y)\right|< C(r,b)(1+|x_1-x_2|)\exp(-\alpha r(y_1-y_2)),\ \ x,y\in\R^2,\ \ y_1-y_2\geq 0,
\ee
where $C$ is continuous on $[a_s/2,a_s)\times S_a$.
\end{proposition}
\begin{proof}
In view of Lemma~\ref{Lem:aux2} and~\eqref{M2as}, it suffices to show 
\be\label{db}
\left|\int_{C_b+ir}dz\rI_2(b;x,y,z)\right|\le C'(r,b)|\rho_2(b;y)||x_1-x_2|\exp(-\alpha r(y_1-y_2)),
\ee
 for all $x,y\in\R^2$ and  $y_1-y_2\geq 0$, where $C'$ is continuous on $[a_s/2,a_s)\times S_a$. (Indeed, combining~\eqref{rE2as}, \eqref{u} and~\eqref{uas}, it is clear that~$|\rE_2^{{\rm as}}
 (b;x,y)|$ is majorized by a continuous function~$c(b)$ for all $(b,x,y)\in S_a\times \R^2\times\R^2$.) Changing variables $z\to z+i(a-b/2+r)$, we rewrite the integral as
\begin{multline}
\rho_2(b;y)\exp(-\alpha r(y_1-y_2))C_2(b;x)^{-1}\\
\times\int_\R dz\exp(i\alpha z(y_1-y_2))\prod_{j=1}^2\frac{G(z+ir-x_j+ia-ib)}{G(z+ir-x_j+ia)}.
\end{multline}
Note that we do not encounter the poles of the $G$-ratios so long as $r\in(0,a_s)$. Furthermore, from \eqref{c} and \eqref{cas}   we obtain the estimate 
\be\label{Gratb}
\left|\frac{G(p+ir+ia-ib)}{G(p+ir+ia)}\right|\leq c(r,b)/\cosh(\ga p),\ \ \ (p,r,b)\in\R\times (0,a_s)\times S_a,
\ee
where
\be\label{gam}
\gamma\equiv \alpha \re b/2=\frac{\pi\re b}{a_{+}a_-},
\ee
and where $c(r,b)$ is continuous on $(0,a_s)\times S_a$. It follows that we have
\be\label{rI2est}
\begin{split}
\left|\int_{C_b+ir}dz\,\rI_2(b;x,y,z)\right|&\leq c(r,b)^2|\rho_2(b;y)|\exp(-\alpha r(y_1-y_2))\\
&\quad\times |C_2(b;x)|^{-1}\int_\R \frac{dz}{\prod_{j=1}^2\cosh(\ga(z-x_j))}.
\end{split}
\ee
By a standard residue calculation, we find that the latter integral equals
\be\label{hypint}
2\frac{x_1-x_2}{\sinh(\ga(x_1-x_2))}.
\ee
(Alternatively, this evaluation can be deduced  from I~Lemma C.1 with $N=1$.) Combining the simple zero of $C_2(x)^{-1}$   along $x_1=x_2$  with the $c$-function asymptotics \eqref{cas}, this yields a bound $|C_2(x)^{-1}/\sinh(\ga(x_1-x_2))|\leq c_1(b)$, with $c_1$ continuous on $S_a$. Hence the desired majorization~\eqref{db} results.
 \end{proof}

In order to generalize the above line of reasoning to the $N=3$ case, we need to obtain a uniform bound on $\rE_2(x,y)$ for $(x,y)\in\C^2\times\R^2$ such that
\be\label{xyres}
v_1-v_2\in(-a_s,0],\ \ \ y_1-y_2\geq 0,\ \ \ \ v=\im x.
\ee
From the pole locations \eqref{rE2ps1}--\eqref{rE2ps2}, it is clear that such a bound is compatible with the poles of $\rE_2(x,y)$. 
In fact, since $\rE_2(x,y)$ has no pole for $v_1-v_2\in (-2a,0]$, one might expect $a_s\to 2a$ in~\eqref{xyres}. However, we are unable to obtain a bound for this larger interval.

 The most obvious starting point would seem to be the representation \eqref{rE2rep}. Now \eqref{CN} and \eqref{cas} entail that the factor $C_2(2a-b;y)^{-1}$ is $O(\exp(\alpha(a-\re b/2)(y_1-y_2))$ as $y_1-y_2\to\infty$. In order to retain boundedness, we need a corresponding damping factor coming from the integral in \eqref{rE2rep}. This can be obtained by shifting the contour $\R$ up to~$C_b$. However, such a shift is only allowed as long as no poles are met. We have already observed that the nearest poles of~$\rI_2$ are located at \eqref{rI2ps}, so this is never the case. As a consequence, we cannot obtain the desired decay factor in any `simple' way.

As it turns out, the representation \eqref{rE2rep2} yields a much better starting point, even though we then have one more term to bound.
It is clear from \eqref{u} and the locations of the $G$-poles I (A.11) that $u(b;x_2-x_1)$ is holomorphic for $-a_s<v_1-v_2<m(\re b)$, where
\be\label{mas}
m(d)\equiv \min(2a-d,d),\ \ \ d\in(0,2a).
\ee
 Using also \eqref{uas} and \eqref{M2as}, we  deduce that for all $(x,y)\in\C^2\times\R^2$ satisfying \eqref{xyres} we have
\be\label{MEb}
|M_2(b;y)\rE_2^{as}(b;x,y)|\leq c(v_1-v_2,b)\exp(-\alpha(y_1v_1+y_2v_2)),
\ee
  where $c$ is continuous on $(-a_s,0]\times S_a$.

Note that $c(v_1-v_2,b)\to\infty$ as $v_1-v_2\downarrow -a_s$, since we then approach the pole of $u(b;x_2-x_1)$ at $x_1-x_2=-ia_s$. Because we prove the bound~\eqref{rE2vb} in the following proposition by using the representation~\eqref{rE2rep2}, we cannot handle the interval $v_1-v_2\in (-2a,-a_s]$.

\begin{proposition}\label{Prop:rE2b}
Letting $(\de,b)\in (0,a_s]\times S_a$, we have
\be\label{rE2vb}
|\rE_2(b;x,y)|<C(\de,b)(1+|\re(x_1-x_2)|)\exp(-\alpha(y_1v_1+y_2v_2))
\ee
for all $(x,y)\in\C^2\times\R^2$ such that
\be\label{N2xyas}
v_1-v_2\in[-a_s+\de,0],\ \ \ y_1-y_2\geq 0,\ \ \ v=\im x,
\ee
where $C$ is continuous on $(0,a_s]\times S_a$. Furthermore, for all $(b,x,y)\in S_a\times \R^2\times \R^2$ we have
\be\label{rE2bex}
|\rE_2(b;x,y)|\le c(b)|x_1-x_2|(1+|y_1-y_2|),
\ee
where $c$ is continuous on $S_a$.
\end{proposition}
\begin{proof}
Choosing first $x\in\R^2$, we begin by rewriting the integral of $\rI_2$  along the $z$-contour $C_b+ir$ in~\eqref{rE2rep2}. Letting $z\to z+x_1+i(a-b/2+r)$, we arrive at 
\begin{multline}
\label{rI2int}
\int_{C_b+ir}dz\,\rI_2(x,y,z)=\rho_2(y)\exp(-\alpha(r-ix_1)(y_1-y_2))C_2(x)^{-1}\\
\times\int_\R dz\exp(i\alpha z(y_1-y_2))\frac{G(z+ir+ia-ib)}{G(z+ir+ia)}\frac{G(z+ir+x_1-x_2+ia-ib)}{G(z+ir+x_1-x_2+ia)}.
\end{multline}
As long as $r\in(0,a_s)$, we stay clear of the poles of the two $G$-ratios. However, when allowing $v_1-v_2< 0$, we must also ensure 
\be
0<r+v_1-v_2<a_s,
\ee
so as not to encounter the poles of the right $G$-ratio for $z+ir+x_1-x_2=0,a_s$. In particular, we can allow any $x\in\C^2$ satisfying $v_1-v_2\in(-a_s,0]$ when we choose (say)
\be
v_1-v_2=-a_s+\de,\ \ \ r=a_s-\de/2,\ \ \ \de\in(0,a_s] . 
\ee
 
The most straightforward way to bound the integral on the right-hand side of \eqref{rI2int} is to estimate the $y$-dependent exponential factor away. Invoking the bound \eqref{Gratb}, this readily yields the estimate 
\begin{multline}\label{rI2intest1}
\left|\int_{C_b+ir}dz\rI_2(x,y,z)\right|\leq c_1(\de,b)^2|\rho_2(y)|\exp(-\alpha(r+v_1)(y_1-y_2)) \\
\times |c(b;x_1-x_2)|^{-1}\int_\R \frac{dz}{\cosh(\ga z)\cosh(\ga(z+\re (x_1-x_2)))},
\end{multline}
with $c_1(\de,b)$ continuous on $(0,a_s]\times S_a$.  We met the latter integral before, cf.~\eqref{rI2est} and~\eqref{hypint}, whence we infer it equals
\be
2\frac{\re(x_1-x_2)}{\sinh(\ga\re(x_1-x_2))}.
\ee

Now $c(b;x_1-x_2)^{-1}$ is regular for $-2a<v_1-v_2<\re b$, vanishes for $x_1-x_2=0$, and has asymptotics
\be\label{crepest}
|c(b;x_1-x_2)^{-1}|\sim C(b)\exp(\gamma |\re (x_1-x_2)|),\ \ \ |\re (x_1-x_2)|\to\infty,
\ee
with $C(b)$ continuous on $S_a$, cf.~\eqref{cas}. Hence we obtain
\be\label{rI2intest2}
\left|\int_{C_b+ir}dz\rI_2(x,y,z)\right|\leq C_1(\de,b)|\rho_2(y)|(1+|\re(x_1-x_2)|)\exp(-\alpha(r+v_1)(y_1-y_2)).
\ee
Combining this with~Lemma~\ref{Lem:aux2}, \eqref{M2as} and \eqref{MEb}, the first assertion now follows.   

To prove the second one, we may restrict attention to the case $y_1-y_2\ge 0$. (Indeed, we can invoke~\eqref{rE2ri} and boundedness of $u(b;z)$ for $(b,z)\in S_a\times\R$ to handle $y_1-y_2<0$.) We can now proceed as before, with $v_1=v_2=0$. Then we also get $\re (x_1-x_2)\to x_1-x_2$ in~\eqref{rI2intest1}--\eqref{crepest}, so in~\eqref{rI2intest2} we may replace the factor $1+|\re (x_1-x_2)|$ by $|x_1-x_2|$. Hence it suffices to prove (cf.~\eqref{rE2rep2})
\be\label{MEest}
|M_2(b;y)\rE_2^{{\rm as}}(b;x,y)|\le c_1(b)(x_1-x_2)(y_1-y_2),\ \ \ x_1-x_2\ge 0,\ \ y_1-y_2\ge 0.
\ee

Recalling \eqref{M2def} and \eqref{rE2as}, we see that \eqref{MEest} amounts to a bound of the form
\be\label{bF}
\Big| \frac{\exp(-\alpha (a-b/2)p)}{c(2a-b;p)}F(b;q,p)\Big| \le c_2(b)qp,\ \ \ q,p\ge 0,
\ee
where
\be\label{Fdef}
F(b;q,p)\equiv \exp(i\alpha qp/2)-u(b;-q)\exp(-i\alpha qp/2).
\ee
Now from \eqref{cas} we have
\be\label{thest}
\Big| \frac{\exp(-\alpha (a-b/2)p)}{c(2a-b;p)}\Big|\le c_3(b)\tanh (p),\ \ \ p\ge 0.
\ee
Also, from $u(b;0)=1$ and the mean value theorem we infer
\be
\re F(q,p)=q(\partial_q\re F)(\theta_1(q),p),\ \ \im F(q,p)=q(\partial_q\im F)(\theta_2(q),p), \ee
where $\theta_j(q)\in [0,q]$, $ j=1,2$.
This readily yields an estimate 
\be
|F(b;q,p)|\le c_4(b)q(1+p),\ \ q,p\ge 0.
\ee
Combining it with \eqref{thest}, we obtain~\eqref{bF}, so that~\eqref{rE2bex} follows.
\end{proof}

The reader may well ask whether the factor $(y_1-y_2)$ in~\eqref{MEest} is necessary, since $F(q,p)$ is obviously bounded. Its necessity can be gleaned from the special cases
\be
F(a_{\de};q,p) =2i \sin(\alpha qp/2), \ \ \de=+,-.
\ee
More precisely, we need the factor $|x_1-x_2|$ in the bound~\eqref{rE2bex} to push through the proof of Theorem~\ref{Thm:rE3as}, so we cannot bound the left-hand side of~\eqref{MEest} simply by a constant, cf.~\eqref{conest}. (To be sure, we believe that $\rE_2(b;x,y)$ with $b\in S_a$ fixed is bounded on $\R^2\times\R^2$, but we have not proved this.)

\section{The step from $N=2$ to $N=3$}\label{Sec3}

\subsection{Invariance properties and a duality relation}\label{Sec31}
We begin this subsection by obtaining the counterpart of Proposition \ref{Prop:J2sym}.
 
\begin{proposition}\label{Prop:J3sym}
For all $(b,x,y)\in\cD_3$~\eqref{cD3} and $\eta\in\C$, we have
\be\label{J3ri}
J_3(b;x,y)=J_3(b;-x,-y),
\ee 
\be\label{J3hom}
\begin{split}
J_3(b;x,y)=& \exp(-i\alpha\eta(y_1+y_2+y_3))J_3(b;(x_1+\eta,x_2+\eta,x_3+\eta),y)\\
& =\exp(-i\alpha\eta(x_1+x_2+x_3))J_3(b;x,(y_1+\eta,y_2+\eta,y_3+\eta)) .
\end{split}
\ee
\end{proposition}
\begin{proof}
Following the proof of Proposition \ref{Prop:J2sym}, we obtain, using \eqref{I3}--\eqref{cS3}, \eqref{J2ri} and the reflection equation I~(A.6),
\be
I_3(-x,-y,-z)=I_3(x,y,z).
\ee
Hence \eqref{J3ri} follows as before. The alternative representation~\eqref{J3cm} entails~\eqref{J3hom}. 
\end{proof}
 
We continue by deducing a new representation for $J_3$ that is related to \eqref{J3} by the involution $(b,x,y)\mapsto(2a-b,y,x)$. Aiming to follow the flow chart of Subsection \ref{Sec21}, we first need  a suitable generalization of the Plancherel relation \eqref{Plancherel}. This involves a generalized Fourier transform with kernel
\begin{multline}\label{rF2}
\rF_2(b;x,y)\equiv (a_+a_-)^{-1/2}G(ib-ia)W_2(b;x)^{1/2}J_2(b;x,y)W_2(2a-b;y)^{1/2},\\
 b\in(0,2a),\ \ x,y\in G_2.
\end{multline}
(Here and below, we choose positive square roots.) For future reference, we note the symmetry properties
\be\label{rF2s}
\rF_2(b;-x,-y)=\rF_2(b;x,y),\ \ \ \rF_2(b;x,y)=\rF_2(2a-b;y,x),
\ee
cf.~Propositions \ref{Prop:J2sym}--\ref{Prop:J2prop}. (Actually $\rF_2(b;x,y)$ is self-dual, too; this can be readily checked by using~\eqref{J2sd}.)

By specialization of results in \cite{R03} (cf.~also Subsection 2.2 in \cite{R11}), we inferred in Section 3 of \cite{HR15} that the operator
\be
\cF_2(b): \cC_2\equiv C_0^\infty(G_2)\subset L^2(G_2)\to L^2(G_2),\ \ \ b\in(0,2a),
\ee
defined by
\be
(\cF_2(b)\psi)(x)\equiv \frac{1}{a_+a_-}\int_{G_2}\rF_2(b;x,y)\psi(y)dy,\ \ \ \psi\in\cC_2,\ \ x\in G_2,
\ee
extends to a unitary operator. Observing that (cf.~I (A.6), (A.9))
\be\label{rF2conj}
\overline{\rF_2(b;x,y)}=\rF_2(b;x,-y),\ \ \ b\in(0,2a),\ \ x,y\in G_2,
\ee
we thus arrive at the generalized Plancherel relation
\be\label{cF2plan}
\int_{G_2}dzf(z)g(z)=\int_{G_2}dp (\cF_2 f)(p)(\cF_2 g)(-p),\ \ \ f,g\in L^2(G_2)\cap L^1(G_2).
\ee

Restriction attention to $b\in(0,2a)$ at first, we choose
\be
f(z)=\cS_2(b;(x_1,x_2),(z_1,z_2))W_2(b;z)^{1/2},
\ee
and
\be
\begin{split}
g(z)&=(a_+a_-)^{1/2}G(ia-ib)\exp(i\alpha y_3(x_1+x_2+x_3))\\
&\quad \times W_2(2a-b;(y_1,y_2))^{-1/2}\rF_2(b;z,(y_1-y_3,y_2-y_3))\prod_{k=1}^2\frac{G(x_3-z_k-ib/2)}{G(x_3-z_k+ib/2)},
\end{split}
\ee
cf.~\eqref{cS2}. Then it follows from \eqref{J3} and \eqref{rF2} that $J_3(b;x,y)$ is given by the left-hand side of \eqref{cF2plan}.   The crux is now that the $\cF_2$-transforms of the functions $f$ and $g$ chosen above can be readily computed by using results from \cite{HR15}. We proceed to embark on this.

From Eq.~(3.18) in \cite{HR15} we recall the integral equation
\be
\int_{G_2}dzW_2(b;t)^{1/2}\cS_2(b;t,z)W_2(b;z)^{1/2}\rF_2(b;z,p)=\mu(b;p)\rF_2(b;t,p),\ \ \ t,p\in G_2,
\ee
with eigenvalue
\be
\mu(b;p)\equiv a_+a_-G(ia-ib)^2\prod_{j=1}^2\prod_{\de=+,-}G(\de p_j-ia+ib/2).
\ee
(This result can be regarded as the $N=2$ counterpart of the $N=1$ formula~\eqref{ftr}.)
Clearly, this implies 
\be\label{cF2f}
(\cF_2 f)(p)=(a_+a_-)^{-1}W_2(b;(x_1,x_2))^{-1/2}\mu(b;p)\rF_2(b;(x_1,x_2),p).
\ee

Using the reflection equation I (A.6), we find 
\be
\prod_{k=1}^2\frac{G(x_3-z_k-ib/2)}{G(x_3-z_k+ib/2)}=(a_+a_-)^{-1}G(ia-ib)^2 \mu(2a-b;(z_1-x_3,z_2-x_3)),
\ee
which yields 
\begin{multline}
(\cF_2 g)(-p)=(a_+a_-)^{-3/2}G(ia-ib)^3\exp(i\alpha y_3(x_1+x_2+x_3))W_2(2a-b;(y_1,y_2))^{-1/2}\\
\times \int_{G_2}dz\mu(2a-b;(z_1-x_3,z_2-x_3))\rF_2(b;z,(y_1-y_3,y_2-y_3))\rF_2(b;z,-p).
\end{multline}
Taking $z_k\to z_k+x_3$, we deduce from \eqref{rF2} and Proposition \ref{Prop:J2sym} that the integral on the right-hand side can be rewritten as
\begin{multline}
\exp(i\alpha x_3(y_1+y_2-2y_3-p_1-p_2))\\
\times \int_{G_2}dz\mu(2a-b;z)\rF_2(b;z,(y_1-y_3,y_2-y_3))\rF_2(b;z,-p).
\end{multline}
Keeping in mind \eqref{rF2s} and \eqref{rF2conj}, we infer from Eq.~(3.24) in \cite{HR15} the generalized eigenfunction expansion
\begin{multline}
W_2(2a-b;q)^{1/2}\cS_2(2a-b;q,p)W_2(2a-b;p)^{1/2}\\
=\frac{1}{(a_+a_-)^2}\int_{G_2}dz\mu(2a-b;z)\rF_2(b;z,q)\rF_2(b;z,-p),\ \ \ q,p\in G_2.
\end{multline}
Hence we arrive at the generalized Fourier transform formula
\begin{multline}\label{cF2g}
(\cF_2 g)(-p)=(a_+a_-)^{1/2}G(ia-ib)^3\exp(i\alpha x_3(y_1+y_2+y_3))\\
\times\exp\big(i\alpha[y_3(x_1+x_2)-x_3(2y_3+p_1+p_2)]\big)\\
\times W_2(2a-b;p)^{1/2}\cS_2(2a-b;(y_1-y_3,y_2-y_3),p).
\end{multline}

Substituting \eqref{cF2f} and \eqref{cF2g} in the right-hand side of \eqref{cF2plan}, taking $p_k\to p_k-y_3$ and rewriting the resulting integral in terms of $J_2$ by using~\eqref{rF2}, we obtain
\begin{multline}
J_3(b;x,y)=G(ia-ib)^4\exp(i\alpha x_3(y_1+y_2+y_3))\exp(i\alpha y_3(x_1+x_2))
\\
\times\int_{G_2}dp\, S_3^{\sharp}(2a-b;y,p)W_2(2a-b;p)\exp(-i\alpha x_3(p_1+p_2))J_2(b;x,(p_1-y_3,p_2-y_3)).
\end{multline}
Using now the $J_2$-duality relation~\eqref{J2drel} and invariance property~\eqref{J2hom}, we deduce  the representation
\be\label{J3d}
J_3(b;x,y)= G(ia-ib)^6\exp(i\alpha x_3(y_1+y_2+y_3))\int_{G_2}dp\, I_3(2a-b;y,x,p).
\ee
We note that this formula is valid for all $b\in S_a$ and $x,y\in\R^3$. We are now prepared for the $N=3$ analog of Proposition \ref{Prop:J2prop}. By contrast to the latter, the following theorem amounts to a substantial novel result, proving some of the conjectures in I~Section~7 for the case $N=3$.

\begin{theorem}\label{Thm:J3prop}
Letting $b\in S_a$ and $x,y\in\R^3$, the duality property~\eqref{J3drel} and symmetry property~\eqref{J3sym} hold true. 
\end{theorem}
\begin{proof}
We obtain \eqref{J3drel} upon comparing the representations \eqref{J3} and \eqref{J3d}.
Just as in the $N=2$ case, we then infer invariance under permutations of the variables $(y_1,y_2,y_3)$ by combining \eqref{J3drel} with the manifest invariance under permutations of the variables  $ (x_1,x_2,x_3)$.
\end{proof}

\subsection{Global meromorphy}\label{Sec32}
We proceed to establish global meromorphy for $J_3(b;x,y)$, following the line of reasoning in Subsection \ref{Sec22} as far as possible. Thus we need again a number of preliminaries.
 First, from Proposition \ref{Prop:J3sym} and Theorem \ref{Thm:J3prop}  the following invariance properties of $\cP_3$ are readily inferred:
\be\label{cP3ri}
\cP_3(b;-x,-y)=\cP_3(b;x,y), \ \ \ \mathrm{(reflection~invariance)},
\ee
\be\label{cP3d}
\cP_3(b;x,y)=G(ia-ib)^6\cP_3(2a-b;y,x),\ \ \ \mathrm{(duality)},
\ee
\be\label{cP3pi}
\cP_3(b;x,y)=\cP_3(b;\sigma x,\tau y),\ \ \ (\sigma,\tau)\in S_3\times S_3,\ \ \ \mathrm{(permutation~invariance)}.
\ee

Second, just as in the $N=2$ case, a key ingredient is an eigenvalue equation for $\cP_3$. It corresponds to the $k=1$ eigenvalue equation I (5.13) for $J_3$, with $\de\in\{+,-\}$ chosen such that $a_{-\de}=a_s$, and with $x,y\to -x,-y$. Invoking the reflection invariance~\eqref{J3ri}, the latter A$\De$E is given by
\be\label{J3ev}
\sum_{j=1}^3V_3(b;\sigma_{1j}x)J_3(b;x+ia_se_j,y)=\sum_{j=1}^3 e_l(-2y_j)J_3(b;x,y),
\ee
where $e_j$, $j=1,2,3$, and  
  $\sigma_{kl}$, $k,l=1,2,3$, denote the standard basis elements in $\C^3$ and the reflection that acts on $x\equiv (x_1,x_2,x_3)$ by interchanging $x_k$ and $x_l$, resp.; moreover, the coefficient function reads
\be
V_3(b;x)\equiv \prod_{m=2,3}\frac{s_l(x_m-x_1-ib)}{s_l(x_m-x_1)}.
\ee

Third, we  define counterparts of~\eqref{ep2}--\eqref{cD2n}: 
\be\label{ep3}
\epsilon_3\equiv a_l/4,
\ee
\be\label{cA3}
\cA_3\equiv \{ x\in\C^3\mid v_j -v_k>-\re b,\ \ 1\le j<k\le 3\},
\ee
\be\label{cA3n}
\cA_3^{(n)}\equiv
\left\{
\begin{array}{ll}
  \{ x\in\C^3\mid |v_j-v_k|<a_s+\re b, \ \ 1\le j<k\le 3\},  & n=1   , \\
   \{ x\in\cA_3\mid v_j-v_k<na_s+\re b, \ \ 1\le j<k\le 3\},  & n=0,2,3,\ldots  ,
\end{array}
\right.
 \ee
 \be\label{D3p}
D_3^{(+)}\equiv \big\{ (b,x)\in S(\epsilon_3)\times\C^3\mid x\in\cA_3\big\}, \ee
\be\label{D3n}
D_3^{(n)}\equiv \big\{ (b,x)\in S(\epsilon_3)\times\C^3\mid x\in\cA_3^{(n)}\big\}, 
\ee
\be\label{cD3p}
\cD_3^{(+)}\equiv \big\{ (b,x,y)\in D_3^{(+)}\times\C^3\mid \max_{1\leq j<k\leq 3}|\im (y_j-y_k)|<\re b\big\},
\ee
\be\label{cD3n}
\cD_3^{(n)}\equiv \big\{ (b,x,y)\in D_3^{(n)}\times\C^3\mid \max_{1\leq j<k\leq 3}|\im (y_j-y_k)|<\re b\big\}.
\ee
Then the counterpart of~\eqref{inclu} is
\be\label{inclu3}
D_3^{(1)}\subset D^r_3,\ \ \ \cD_3^{(1)}\subset \cD_3,
\ee
cf.~\eqref{D3}--\eqref{cD3}. To verify these inclusions, we need only note
\be
|\im \tilde{x}_j|\le \frac{1}{3}|v_j-v_k|+\frac{1}{3}|v_j-v_l|,\ \ \ \{j,k,l\}=\{1,2,3\},
\ee
and use 
\be
  \frac{2}{3}(a_s+\re b)<  a-\frac{1}{2}\re b,\ \ \   b\in S(\epsilon_3).
\ee 
    
Now we have the following analog of Lemma \ref{Lemma:cP2eigeeq}.

\begin{lemma}
Letting $(b,x,y)\in\cD_3^{(0)}$, we have the eigenvalue equation
 \be\label{cP3eigeq}
\sum_{j=1}^3\cV_3(b;\sigma_{1j}x)\cP_3(b;x+ia_se_j,y)=\sum_{j=1}^3 e_l(-2y_j)\cP_3(b;x,y),
\ee
 where the coefficient function is given by
\be\label{cVN3}
\begin{split}
\cV_3(b;x) &\equiv -\pi^2\prod_{m=2,3}\frac{\exp(i(2x_m-2x_1-ia_s)K_l)}{s_l(x_m-x_1)}\\
&\quad\times\left[\Gamma\left(\frac{i}{a_l}(x_m-x_1-ib)\right)\Gamma\left(\frac{i}{a_l}(x_m-x_1+ib-2ia)\right)\right]^{-1},
\end{split}
\ee
with $K_l$ defined by \eqref{slKl}.
\end{lemma}
\begin{proof}
The restriction to~$\cD_3^{(0)}$ implies that the four arguments of~$J_3$ occurring in \eqref{J3ev} belong to $\cD_3^{(1)}$, and thus to the holomorphy domain~$\cD_3$, cf.~\eqref{inclu3}. Hence the A$\De$E is well defined. Its similarity transform~\eqref{cP3eigeq} follows from a computation paralleling the one in the proof of Lemma \ref{Lemma:cP2eigeeq}. 
\end{proof}

We are now prepared for   
the following counterpart of Proposition \ref{Prop:cP2ext}, which again proves a conjecture made in I~Section~7. 

\begin{theorem}\label{Thm:cP3ext}
The product function $\cP_3(b;x,y)$~\eqref{cP3} admits a holomorphic continuation from $\cD_3$~\eqref{cD3} to $ S_a\times\C^3\times\C^3$.
\end{theorem}
\begin{proof} For transparency, we follow the reasoning in the proof of Proposition \ref{Prop:cP2ext} as far as possible, even though we need to enlarge on it shortly. Thus,  
we first aim to prove holomorphic continuation to~$\cD_3^{(+)}$~\eqref{cD3p}. Accordingly, 
we assume inductively that $\cP_3(b;x,y)$ is holomorphic in~$\cD_3^{(n)}$ with $n\geq 1$. (This is true for $n=1$, cf.~\eqref{inclu3}.)
To handle the holomorphic continuation to~$\cD_3^{(n+1)}$, 
we begin by rewriting~\eqref{cP3eigeq}. 

First, we introduce
\be\label{cVfact}
\hat{\cV}_3(b;x)=\cV_3(b;x)\prod_{m=2,3}s_l(x_m-x_1).
\ee
Then we multiply \eqref{cP3eigeq} by the two $s_l$-functions, rearrange the terms,  and invoke the permutation invariance \eqref{cP3pi} to obtain 
\begin{multline}\label{cP3id}
\hat{\cV}_3(b;x)\cP_3(b;x+ia_s e_1 ,y)=s_l(x_2-x_1)s_l(x_3-x_1) \sum_{j=1}^3 e_l(-2y_j)\cP_3(b;x,y)\\
+\frac{1}{s_l(x_3-x_2)}\Big[s_l(x_3-x_1)\hat{\cV}_3(b;\sigma_{12}x)\cP_3(b;x+ia_s e_2,y)\\
-s_l(x_2-x_1)\hat{\cV}_3(b;\sigma_{13}x)\cP_3(b;\sigma_{23}(x+ia_s e_3),y)\Big].
\end{multline}
It now follows as before that the multiplier~$\hat{\cV}_3(b;x)$ on the left-hand side is nonzero on~$D_3^{(+)}$, cf.~\eqref{gamzeros}.
 
It is at this point, however, that we can no longer proceed as in the $N=2$ case. For one thing, the zero of the denominator function $s_l(x_3-x_2)$ in~\eqref{cP3id} for $x_3=x_2$ is innocuous (as the bracketed function then vanishes, too), but we need to steer clear of the remaining zeros. 

We can avoid this snag (and other ones) as follows. First, we define domains
\be\label{cA31}
\cA_{3,1}^{(n)}\equiv \{ x\in\cA_3^{(n)}\mid |v_2-v_3|<\re b\},\ \ \ n\ge 1,
\ee
\be\label{D31n}
D_{3,1}^{(n)}\equiv \big\{ (b,x)\in S(\epsilon_3)\times\C^3\mid x\in\cA_{3,1}^{(n)}\big\}, 
\ee
 \be\label{cD31n}
\cD_{3,1}^{(n)}\equiv \big\{ (b,x,y)\in D_{3,1}^{(n)}\times\C^3\mid \max_{1\leq j<k\leq 3}|\im (y_j-y_k)|<\re b\big\}.
\ee
Second, we consider the function~$R_3(b;x,y)$ on the right-hand side of~\eqref{cP3id} for all points $(b,x,y)\in \cD_{3,1}^{(n)}$ such that
\be\label{vres3}
v_1-v_2, v_1-v_3\in ((n-1)a_s-\re b,na_s+\re b).
\ee
This yields a domain $\cD_{3,1,r}^{(n)}\subset\cD_{3,1}^{(n)}$ on which $R_3(b;x,y)$ is holomorphic for $n=1$. Using the induction assumption, we also infer holomorphy for $n>1$. (Note that we need the interchange $\sigma_{23}$ for this to follow.)  The $x$-translation of~$\cD_{3,1,r}^{(n)}$ over $ia_s e_1$ equals $\cD_{3,1,r}^{(n+1)}$, and the latter domain meets~$\cD_3^{(n)}$ for all points with
\be
v_1-v_2, v_1-v_3\in (na_s-\re b,na_s+\re b),\ \ \ |v_2-v_3|<\re b.
\ee

As a result, we obtain a holomorphic continuation of $\cP_3(b;x,y)$ to all of~$\cD_{3,1}^{(n+1)}$. However, this is a proper subdomain of~$\cD_3^{(n+1)}$, so we need yet another enlargement.
This consists in further domains
\be\label{cA33}
\cA_{3,3}^{(n)}\equiv \{ x\in\cA_3^{(n)}\mid |v_1-v_2|<\re b\},\ \ \ n\ge 1,
\ee
\be\label{D33n}
D_{3,3}^{(n)}\equiv \big\{ (b,x)\in S(\epsilon_3)\times\C^3\mid x\in\cA_{3,3}^{(n)}\big\}, 
\ee
 \be\label{cD33n}
\cD_{3,3}^{(n)}\equiv \big\{ (b,x,y)\in D_{3,3}^{(n)}\times\C^3\mid \max_{1\leq j<k\leq 3}|\im (y_j-y_k)|<\re b\big\}.
\ee

Consider now the involution
\be
\varphi: \C^3\to \C^3,\ \ \ x\mapsto -\sigma_{13}x.
\ee
It is easy to check
\be
\varphi (\cA_3^{(n)})=\cA_3^{(n)},\ \ \ \varphi (\cA_{3,1}^{(n)})=\cA_{3,3}^{(n)},\ \ \ n\ge 1,
\ee
so it gives rise to a bijection between the domains~\eqref{cA33}--\eqref{cD33n} and~\eqref{cA31}--\eqref{cD31n}. 

The point is that the invariance properties \eqref{cP3ri}--\eqref{cP3pi} are preserved under analytic continuation, so that we have
\be
\cP_3(b;x,y)=\cP_3(b;\varphi(x),-y),\ \ \ x\in\cD_{3,3}^{(n)}.
\ee
As a consequence, the function  $\cP_3(b;x,y)$ has a holomorphic continuation to~$\cD_{3,1}^{(n+1)}$ as well as to~$\cD_{3,3}^{(n+1)}$. 

The latter two domains are tube domains with open, connected bases, and the two bases have a nontrivial intersection. By Bochner's Theorem~1.1 it then follows that $\cP_3(b;x,y)$ has a holomorphic continuation to the tube whose base is the convex hull of the latter two bases. We claim that this tube equals~$\cD_3^{(n+1)}$. Taking this claim for granted, we have completed the induction argument, so it follows that~$\cP_3$ continues to~$\cD_3^{(+)}$. 

Now we need only invoke $S_3$-symmetry in~$x$ to obtain holomorphy of $\cP_3(b;x,y)$ in the tube with base
\be\label{cBep3}
\cB_3(\epsilon_3)\equiv \{(\re b, \im x,\im y)\in (0,\epsilon_3)\times\R^3\times \R^3\mid \max_{1\leq j<k\leq 3}|\im (y_j-y_k)|< \re b\} .
\ee
Then we are in the position to follow again the reasoning for the $N=2$ case, with the equations~\eqref{cB2}--\eqref{cB2h} all having  $N=3$ counterparts that will be clear upon comparing~\eqref{cBep3} with~\eqref{cBep2}. 

To conclude the proof of the theorem, it remains to prove the claim. We can reduce this to a claim for a set $U$ of two real numbers $u_1\equiv v_1-v_2 , u_2\equiv v_2-v_3$ satisfying
\be
u_1,u_2,u_1+u_2\in(-c,d),\ \ 0<c<d.
\ee
Specifically, the claim now amounts to the convex set $U$ being equal to the convex hull of its two convex subsets
\be
U_j\equiv \{ (u_1,u_2)\in U\mid u_j\in (-c,c)\},\ \ j=1,2.
\ee
Rephrased this way, a moment's thought suffices to establish the validity of the claim. (Any $u\in U$ that is not in $U_1\cup U_2$ belongs to the interior of the triangle with corners $(0,0), (d,0), (0,d)$, and $(d,0)$/$(0,d)$   belongs to the closure of $U_2/U_1$.) Hence the theorem follows.
\end{proof}

\subsection{Asymptotics}\label{Sec33}
Introducing the function
\be\label{defd}
d_3(y)\equiv\min_{1\leq j<k\leq 3}(y_j-y_k),\ \ \ y\in\R^3,
\ee
we proceed to elucidate the asymptotic behavior of the function $\rE_3(b;x,y)$ \eqref{rE3} for $d_3(y)\to\infty$.

Combining \eqref{c}--\eqref{CN} with \eqref{GE} and Theorem \ref{Thm:cP3ext}, we find that $\rE_3(b;x,y)$, $b\in S_a$, is meromorphic in $x$ and $y$, with $b$-independent poles located at
\be
z_j-z_k=-2ia-ip_{mn},\ \ \ z=x,y,\ \ 1\leq j<k\leq 3,\ \ m,n\in\N,
\ee
and $b$-dependent pole locations
\be
z_j-z_k=ib+ip_{mn},\ \ z_j-z_k=2ia-ib+ip_{mn},\ \ \ z=x,y,\ \ 1\leq j<k\leq 3,\ \ m,n\in\N.
\ee
Just as in the $N=2$ case, we now assemble further features in a lemma.

\begin{lemma}\label{Lemma:rE3prop}
For all $(b,x,y)\in S_a\times\C^3\times\C^3$ and $\eta\in\C$, the function $\rE_3(b;x,y)$~\eqref{rE3} satisfies
\be\label{rE3ri}
\rE_3(b;-x,-y)=\rE_3(b;x,y)\prod_{1\le j<k\le 3}u(b;x_j-x_k)u(b;y_j-y_k), 
\ee
\be\label{rE3hom}
\begin{split}
\rE_3(b;x,y)=& \exp(-i\alpha\eta(y_1+y_2+y_3))\rE_2(b;(x_1+\eta,x_2+\eta,x_3+\eta),y)\\
& =\exp(-i\alpha\eta(x_1+x_2+x_3))\rE_2(b;(x,(y_1+\eta,y_2+\eta,y_3+\eta)),
\end{split}
\ee
\be\label{rE3d}
\rE_3(b;x,y)=\rE_3(2a-b;y,x),
\ee
\be\label{rE3p}
\rE_3(\sigma x,\tau y)=\rE_3(x,y)\prod_{\substack{j<k\\ \sigma^{-1}(j)>\sigma^{-1}(k)}}(-u(x_j-x_k))\prod_{\substack{j<k\\ \tau^{-1}(j)>\tau^{-1}(k)}}(-u(y_j-y_k)),\ \ \ (\sigma,\tau)\in S_3\times S_3, 
\ee
where, e.~g., $(\sigma x)_j\equiv x_{\sigma(j)},\  j=1,2,3$. 
\end{lemma}
\begin{proof}
Like in the $N=2$ case, these properties are easily derived from the corresponding features of~$J_3(b;x,y)$ in Proposition~\ref{Prop:J3sym} and  Theorem~\ref{Thm:J3prop}.
\end{proof}

Recalling from I (2.11) the kernel function
\be\label{cK3}
\cK_3^\sharp(b;x,z)\equiv [C_3(b;x)C_2(b;-z)]^{-1}\cS_3^\sharp(b;x,z),
\ee
we infer from \eqref{rE2}--\eqref{J3} and \eqref{rE3} the representation
\begin{multline}\label{rE3rep}
\rE_3(b;x,y)=\frac{(\phi(b)G(ib-ia))^2}{2a_+a_-}\\
\times\frac{\exp(i\alpha y_3(x_1+x_2+x_3))}{\prod_{n=1}^2c(2a-b;y_n-y_3)}\int_{\R^2}dz\,\rI_3(b;x,y,z),\ \ \ b\in S_a,\ \ x,y\in\R^3,
\end{multline}
where the integrand is given by
\be\label{rI3}
\rI_3(b;x,y,z)=\cK_3^\sharp(b;x,z)\rE_2(b;z,(y_1-y_3,y_2-y_3)).
\ee
Indeed, since the integrand~$I_3$ in~\eqref{J3} is clearly invariant under the interchange $z_1\leftrightarrow z_2$, we can replace the integration  over the Weyl chamber~$G_2$ in~\eqref{J3} by integration over~$\R^2$  times 1/2.

Following the $N=2$ case, we deduce the dominant asymptotics of $E_3$ by shifting the $z_k$-contours $\R$ in \eqref{rE3rep} up past the poles of $\rI_3$ located at
\be\label{ps}
z_k=x_j+ia-ib/2,\ \ \ k=1,2,\ \ j=1,2,3.
\ee
Recalling the $G$-zeros \eqref{Ezs}, we infer from \eqref{rE3} and \eqref{c}--\eqref{CN} that $\rE_3$ vanishes along the hyperplanes $x_j=x_k$, $1\leq j<k\leq 3$. Hence we may as well require
\be\label{diffx}
x_j\neq x_k,\ \ \ 1\leq j<k\leq 3,
\ee
so that the poles \eqref{ps} are simple. 

In order to keep track of the residues appearing, we need to shift the two contours separately. Assuming first $\im(z_1-z_2)\in(-a_s,0]$, we note that Proposition \ref{Prop:rE2b} and the bounds \eqref{cas}, I (B.6)   entail that the integrand $\rI_3$ has exponential decay for $|\re z_k|\to\infty$. Moreover, from invariance of $\rI_3$   under   $z_1\leftrightarrow z_2$ it follows that $\rI_3$ has the same decay for $\im(z_1-z_2)\in[0,a_s)$. Hence, as long as the contours are separated by a distance less than $a_s$, we encounter no problems with the contour tails. We must, however, take care to avoid the $x_j$-independent poles of~$\rI_3$, which are due either to zeros of $C_2(-z)$ or poles of $E_2(z,(y_1-y_3,y_2-y_3))$. The former are located at
\be
z_1-z_2=2ia+ip_{kl},\ \ z_1-z_2=-ib-ip_{kl},\ \ \ k,l\in\N,
\ee
whereas the locations of the latter are given by \eqref{rE2ps1}--\eqref{rE2ps2}. Recalling the function $m(d)$ \eqref{mas}, we thus see that 
 the poles in question are not met for $|\im(z_1-z_2)|<m(\re b)$.

Next, we let $x(\nu)$, $\nu=1,2,3$, denote the variables obtained by removing $x_\nu$ from $x\equiv (x_1,x_2,x_3)$:
\be
x(1)=(x_2,x_3),\ \ \ x(2)=(x_1,x_3),\ \ \ x(3)=(x_1,x_2).
\ee
%
%
Introducing the functions
\be\label{M3}
M_3(b;y)\equiv \frac{\phi(b)^2}{\prod_{n=1}^2c(2a-b;y_n-y_3)}\rho_3(b;y),
\ee
\be
\rho_3(b;y)\equiv \exp(-\alpha(a-b/2)(y_1+y_2-2y_3)),
\ee
we are prepared for the following counterpart of Lemma~\ref{Lem:aux2}.

\begin{lemma}\label{Lem:aux3}
Letting $(r,b)\in(0,a_s)\times S_a$ and $x,y\in\R^3$ with the $x$-restriction \eqref{diffx} in effect, we have
\begin{multline}\label{rE3rep2}
\frac{\rE_3(x,y)}{M_3(y)}\exp(-i\alpha y_3(x_1+x_2+x_3))\\
=\frac{1}{\rho_3(b;y)} \Bigg[\frac{G(ib-ia)^2}{2a_+a_-}\int_{(C_b+ir)^2}dz\,\rI_3(x,y,z)\\
+\frac{G(ib-ia)}{\sqrt{a_+a_-}}\sum_{\nu=1}^3\prod_{j<\nu}(-u(x_\nu-x_j))\cdot \int_{C_b+ir}dt\,\hat{\rI}_{3,\nu}(x,y,t)\Bigg]\\
+\sum_{\nu=1}^3\frac{C_3(x(\nu),x_\nu)}{C_3(x)}\rE_2(x(\nu),(y_1-y_3,y_2-y_3)),
%
\end{multline}
with
\be\label{hI3}
\hat{\rI}_{3,\nu}(b;x,y,t)\equiv \cK_2^\sharp(b;x(\nu),t)\rE_2(b;(x_\nu+ia-ib/2,t),(y_1-y_3,y_2-y_3)),
\ee
where $\cK_2^\sharp$ is given by~\eqref{cK2} and $C_b$ by~\eqref{Cb}.
\end{lemma}

\begin{proof}
First, we note that by \eqref{rE3rep}--\eqref{rI3} and \eqref{M3} the left-hand side of \eqref{rE3rep2} equals
\be\label{rE3repLHS}
\frac{\cG^2}{2\rho_3(y)}\int_{\R^2}dz\, \cK_3^\sharp(x,z)\rE_2(z,\hat{y}),
\ee
where we have introduced
\be
\hat{y}\equiv (y_1-y_3,y_2-y_3),\ \ \ \cG\equiv \frac{G(ib-ia)}{\sqrt{a_+a_-}}.
\ee
When determining the effect of the pertinent contour shifts, we find it convenient to work with $J_2(z,(y_1-y_3,y_2-y_3))$, since it is invariant under the interchange $z_1\leftrightarrow z_2$. Therefore, we use \eqref{rE2} and \eqref{cK3} to rewrite \eqref{rE3repLHS} as
\be\label{rE3repLHS2}
\frac{\phi(b)\cG^3}{2\rho_3(b;y)}\frac{\cL_3(b;x,y)}{C_3(b;x)C_2(2a-b;\hat{y})},
\ee
with
\be
\cL_3(b;x,y)\equiv \int_{\R^2}dz\, W_2(b;z)\cS^\sharp_3(b;x,z)J_2(b;z,\hat{y}).
\ee

Letting
\be
0<\ep<\min(m(\re b)/2,a_s/2),
\ee
we move the two contours $\R$ simultaneously up to $C_b-i\ep$ without meeting poles. Moreover, shifting the $z_1$-contour up by a further amount $2\ep$, we only encounter the three simple poles \eqref{ps} with $k=1$. These poles are due to the factor $G(x_j-z_1-ib/2)$ in $\cS_3^\sharp(x,z)$ \eqref{cS3}, and the $G$-residue I (A.13) entails
\be
\lim_{z_1\to x_j+ia-ib/2}(z_1-x_j-ia+ib/2)G(x_j-z_1-ib/2)=\lim_{z_1\to -ia}(-z-ia)G(z)=\frac{\sqrt{a_+a_-}}{2\pi i}.
\ee
Observing that
\be
2\pi i\, \Res\frac{G(x_j-z_1-ib/2)}{G(x_j-z_1+ib/2)}\Big\arrowvert_{z_1=x_j+ia-ib/2}=\frac{\sqrt{a_+a_-}}{G(ib-ia)}=\cG^{-1},
\ee
we thus deduce
\begin{multline}\label{cL3Expr}
\cL_3(x,y)=\int_{C_b+i\epsilon}dz_1\int_{C_b-i\epsilon}dz_2\, W_2(z)\cS^\sharp_3(x,z)J_2(z,\hat{y})\\
+\cG^{-1}\int_{C_b-i\epsilon}dz_2\, \sum_{\nu=1}^3\cR_\nu(x,z_2)J_2((x_\nu+ia-ib/2,z_2),\hat{y}),
\end{multline}
with remainder residue (cf.~\eqref{c})
\be\label{Rnu}
\begin{split}
\cR_\nu(x,z_2) &= \frac{1}{c(x_\nu-z_2+ia-ib/2)c(z_2-x_\nu-ia+ib/2)}\\
&\quad \times \prod_{j=1}^3 c(z_2-x_j-ia+ib/2)\cdot \prod_{\substack{j=1\\ j\neq\nu}}^3 c(x_\nu-x_j)\\
&= \prod_{\substack{j=1\\ j\neq\nu}}^3 c(z_2-x_j-ia+ib/2)\cdot \frac{\prod_{\substack{j=1\\ j\neq\nu}}^3 c(x_\nu-x_j)}{c(x_\nu-z_2+ia-ib/2)}\\
&= \cS_2^\sharp(x(\nu),z_2)\frac{\prod_{\substack{j=1\\ j\neq\nu}}^3 c(x_\nu-x_j)}{c(x_\nu-z_2+ia-ib/2)}.
\end{split}
\ee

Now shifting the $z_2$-contours in \eqref{cL3Expr} up by $2\ep$, we only encounter the poles \eqref{ps} with $k=2$. In the residues spawned by the first integral we perform the interchange $z_1\leftrightarrow z_2$ and use the corresponding invariance of $J_2(z,\hat{y})$ to get
\be\label{IntSum}
\int_{(C_b+i\epsilon)^2}dz\, W_2(z)\cS_3^\sharp(x,z)J_2(z,\hat{y})+\cG^{-1}\int_{C_b+i\epsilon}dz_2\, \sum_{\nu=1}^3\cR_\nu(x,z_2)J_2((x_\nu+ia-ib/2,z_2),\hat{y}).
\ee
The second integral in \eqref{cL3Expr} yields a copy of the second integral in \eqref{IntSum} plus the residue term
\be
\cG^{-2}\sum_{\substack{\nu_1,\nu_2=1\\ \nu_1\neq\nu_2}}^3 \cR_{\nu_1,\nu_2}(x)J_2((x_{\nu_1}+ia-ib/2,x_{\nu_2}+ia-ib/2),\hat{y}),
\ee
where
\be
\cR_{\nu_1,\nu_2}(x)=\prod_{\ell=1}^2c(x_{\nu_\ell}-x_{\nu_3}),\ \ \ \{\nu_1,\nu_2,\nu_3\}=\{1,2,3\}.
\ee
Hence, using invariance under the interchange $x_{\nu_1}\leftrightarrow x_{\nu_2}$, we obtain
\begin{multline}\label{cL3Expr2}
\cL_3(x,y)=\int_{(C_b+i\epsilon)^2}dz\, W_2(z)\cS_3^\sharp(x,z)J_2(z,\hat{y})\\
+2\cG^{-1}\int_{C_b+i\epsilon}dt\, \sum_{\nu=1}^3\cR_\nu(x,t)J_2((x_\nu+ia-ib/2,t),\hat{y})\\
+2\cG^{-2}\sum_{1\leq \nu_1<\nu_2\leq 3} \cR_{\nu_1,\nu_2}(x)J_2((x_{\nu_1}+ia-ib/2,x_{\nu_2}+ia-ib/2),\hat{y}).
\end{multline}

Shifting all contours up to $C_b+ir$ without encountering further poles, we proceed to reformulate the resulting expression in terms of $\rE_2$. From \eqref{Rnu}, \eqref{cK2} and \eqref{rE2}, we infer
\be\label{ratio}
\begin{split}
\frac{\cR_\nu(x,t)J_2((x_\nu+ia-ib/2,t),\hat{y})}{C_2(2a-b;\hat{y})}&=(\phi(b)\cG)^{-1}\rE_2((x_\nu+ia-ib/2,t),\hat{y})\\
&\quad \times \cK_2^\sharp(x(\nu),t)C_2(x(\nu))\prod_{\substack{j=1\\ j\neq\nu}}^3 c(x_\nu-x_j).
\end{split}
\ee
Multiplying \eqref{cL3Expr2} by the prefactors in \eqref{rE3repLHS2}, writing
\be
C_3(x)=C_2(x(\nu))\prod_{j<\nu}c(x_j-x_\nu)\cdot \prod_{j>\nu}c(x_\nu-x_j)
\ee
and using \eqref{ratio}, \eqref{u} and \eqref{rE2hom}, we arrive at the right-hand side of \eqref{rE3rep2}.
\end{proof}

Multiplying \eqref{rE3rep2} by $M_3(y)\exp(i\alpha y_3(x_1+x_2+x_3))$, we continue by analyzing the last sum in the resulting expression, anticipating that it yields the dominant asymptotics of $\rE_3$. Using~\eqref{phuinv} and~\eqref{cas}, we readily deduce
\be\label{M3as}
|M_3(b;y)-1|\le  c(b, \rho)\exp(-\alpha\rho d_3(y)),\ \ \ (b, y,\rho)\in S_a\times \R^3\times  [a_s/2,a_s),\ \ \  d_3(y)\ge 0,
\ee
where $c(b, \rho)$ is continuous on $S_a\times [a_s/2,a_s)$.
Moreover, observing that the function $\rE_2^{\rm as}(z,w)$ \eqref{rE2as} can be rewritten
\be
\rE_2^{\rm as}(z,w)=\sum_{\tau\in S_2}\frac{C_2(z_\tau)}{C_2(z)}\exp(i\alpha z_\tau\cdot w),
\ee
we infer from Proposition \ref{Prop:rE2as} that
\begin{multline}
\exp(i\alpha y_3(x_1+x_2+x_3))\rE_2(x(\nu),(y_1-y_3,y_2-y_3))\\
=\sum_{\substack{\sigma\in S_3\\ \sigma(3)=\nu}}\frac{C_2(x_{\sigma(1)},x_{\sigma(2)})}{C_2(x(\nu))}\exp(i\alpha x_\sigma\cdot y)+R_\nu(x,y),
\end{multline}
with the remainder satisfying
\be\label{Rnub}
|R_\nu(b;x,y)|\le C(r,b)(1+|x(\nu)_1-x(\nu)_2|)\exp(-\alpha r(y_1-y_2)),
\ee
for all $(b,x,y)\in S_a\times \R^3\times \R^3$ with~$y_1-y_2\ge 0$. Due to the identity
\be
\frac{C_3(x(\nu),x_\nu)C_2(x_{\sigma(1)},x_{\sigma(2)})}{C_2(x(\nu))}=c(x_{\sigma(1)}-x_{\sigma(2)})\prod_{\substack{j=1\\ j\neq\nu}}^3 c(x_j-x_{\sigma(3)})=C_3(x_\sigma),
\ee
we thus have
\begin{multline}\label{rE3Exp}
\exp(i\alpha y_3(x_1+x_2+x_3))\sum_{\nu=1}^3\frac{C_3(x(\nu),x_\nu)}{C_3(x)}\rE_2(x(\nu),(y_1-y_3,y_2-y_3))\\
=\sum_{\sigma\in S_3}\frac{C_3(x_\sigma)}{C_3(x)}\exp(i\alpha x_\sigma\cdot y)+R(x,y)=\rE_3^{\rm as}(b;x,y)+R(x,y),
\end{multline}
with remainder
\be
R(b;x,y)=\sum_{\nu=1}^3\frac{C_3(x(\nu),x_\nu)}{C_3(x)}R_\nu(b;x,y).
\ee
Combining \eqref{CN} and the $c$-function asymptotics \eqref{cas} with the bound \eqref{Rnub}, we obtain the majorization
\be\label{Rb}
|R(b;x,y)|\leq C(r,b)\sum_{1\leq j<k\leq 3}(1+|x_j-x_k|)\cdot \exp(-\alpha r(y_1-y_2)),
\ee
valid for all $(b,x,y)\in S_a\times \R^3\times \R^3$ with~$y_1-y_2\ge 0$, and with $C$  continuous on $[a_s/2,a_s)\times S_a$.

Our considerations thus far suggest that the dominant asymptotics of $\rE_3$ is given by \eqref{E3sc}. The following counterpart of Proposition \ref{Prop:rE2as} substantiates this, together with a crucial remainder estimate.

\begin{theorem}\label{Thm:rE3as}
Letting $(r,b)\in  [a_s/2,a_s)\times S_a$, we have
\be
\left|\left(\rE_3-\rE_3^{{\rm as}}\right)(b;x,y)\right|<C(r,b)\prod_{1\leq j<k\leq 3}(1+|x_j-x_k|)\cdot\exp(-\alpha rd_3(y)), 
\ee
for all $x,y\in \R^3$ with $d_3(y)> 0$; here, $C$ is continuous on $ [a_s/2,a_s)\times S_a$.
\end{theorem}

\begin{proof}
It follows from Lemma~\ref{Lem:aux3}  and \eqref{M3as}, \eqref{rE3Exp} and \eqref{Rb} that it suffices to prove the bounds
\be\label{rI3bou}
\Big|\int_{(C_b+ir')^2}dz\,\rI_3(x,y,z)\Big|<C_0(r,b)|\rho_3(b;y)|\prod_{1\leq j<k\leq 3}|x_j-x_k|\cdot\exp(-\alpha rd_3(y)),
\ee
\begin{multline}\label{rI3mbou}
\Big|  \int_{C_b+ir'}dt\,\hat{\rI}_{3,\nu}(x,y,t)\Big|<C_\nu(r,b)|\rho_3(b;y)|\\ \times\big(1+|x(\nu)_1-x(\nu)_2|\big)\Big(1+\sum_{j=1}^2|x_\nu-x(\nu)_j|\Big)\exp(-\alpha rd_3(y)),\ \ \ \nu=1,2,3,
\end{multline}
for all $x,y\in\R^3$ with $d_3(y)> 0$. Here we have introduced
\be\label{rp}
r'\equiv (r+a_s)/2\in(r,a_s),
\ee
and the functions $C_0,\ldots,C_3$ are continuous  on $ [a_s/2,a_s)\times S_a$.

Taking $z_k\to z_k+i(a-b/2+r')$, we use the identity \eqref{rE2hom} to deduce
\begin{multline}\label{rI3Expr}
\int_{(C_b+ir')^2}dz\,\rI_3(x,y,z)= \rho_3(b;y)\exp(-\alpha r'(y_1+y_2-2y_3))\\
\times C_3(b;x)^{-1}\int_{\R^2}dz\, \frac{\rE_2(b;z,(y_1-y_3,y_2-y_3))}{c(b;z_2-z_1)}\prod_{j=1}^3\prod_{k=1}^2\frac{G(z_k+ir'-x_j+ia-ib)}{G(z_k+ir'-x_j+ia)}.
\end{multline}
Next, we note that \eqref{c} and \eqref{cas} imply
\be\label{crepb}
|c(b;z)^{-1}|\le C(b)|\sinh(\gamma z)|,\ \ (b,z)\in S_a\times \R,\ \ \gamma=\alpha \re b/2,
\ee
with $C$ continuous on $S_a$. Combining this with the estimates~\eqref{rE2bex} and  \eqref{Gratb}, we deduce
\begin{multline}\label{conest}
\left|\int_{(C_b+ir')^2}dz\,\rI_3(x,y,z)\right|\leq c_2(r,b)|\rho_3(b;y)|\exp(-\alpha r'(y_1+y_2-2y_3))(1+y_1-y_2)\\
\times |C_3(b;x)|^{-1}\int_{\R^2}dz\,\frac{(z_1-z_2)\sinh(\ga(z_1-z_2))}{\prod_{j=1}^3\prod_{k=1}^2\cosh(\ga(x_j-z_k))},\  \ x,y\in\R^3,\ \ d_3(y)>0,
\end{multline}
for some $c_2$ continuous on $[a_s/2,a_s)\times S_a$. An explicit evaluation of the integral on the right-hand side can be obtained from the $N=2$ case of I Lemma C.2, which  yields
\be\label{intEval}
\int_{\R^2}dz\,\frac{(z_1-z_2)\sinh(\ga(z_1-z_2))}{\prod_{j=1}^3\prod_{k=1}^2\cosh(\ga(x_j-z_k))}=4\ga^{-3}\prod_{1\leq j<k\leq 3}\frac{\ga(x_j-x_k)}{\sinh(\ga(x_j-x_k))}.
\ee
Now we use~\eqref{crepb} once more to obtain
\be
\Big|C_3(b;x)^{-1}\Big/\prod_{1\leq j<k\leq 3}\sinh(\ga(x_j-x_k))\Big|\leq c_3(b),
\ee
with $c_3$ continuous on $S_a$. Finally, since we assume $d_3(y)$~\eqref{defd} is positive, we have
\begin{multline}
(y_1-y_2)\exp(-\alpha r'(y_1+y_2-2y_3)) 
<(y_1-y_3)\exp(-\alpha r'((y_1-y_3)+(y_2-y_3)))\\ <C(r)\exp(-\alpha r((y_1-y_3)+(y_2-y_3)))<C(r)\exp(-\alpha rd_3(y)),
\end{multline} 
with $C$ continuous on $[a_s/2,a_s)$. Putting the pieces together, the desired majorization~\eqref{rI3bou}  easily follows.

We continue by proving~\eqref{rI3mbou}. Taking $t\to t+i(a-b/2+r')$ and appealing once more to \eqref{rE2hom}, we arrive at 
\begin{multline}\label{rI3mExpr}
\int_{C_b+ir'}dt\,\hat{\rI}_{3,\nu}(x,y,t)= \rho_3(b;y)C_2(b;x(\nu))^{-1}\\
\times \int_\R dt\,\rE_2(b;(x_\nu,t+ir'),(y_1-y_3,y_2-y_3))\prod_{j\neq\nu}\frac{G(t+ir'-x_j+ia-ib)}{G(t+ir'-x_j+ia)}.
\end{multline}
Using Proposition \ref{Prop:rE2b} and the bounds \eqref{cas}, \eqref{Gratb}, we now deduce
\begin{multline}\label{hI3b}
\left|\int_{C_b+ir'}dt\,\hat{\rI}_{3,\nu}(x,y,t)\right|\leq c_4(r,b)|\rho_3(b;y)|\exp(-\alpha r'(y_2-y_3))\\
\times \exp(\ga|x(\nu)_1-x(\nu)_2|)\int_\R dt\,(1+|x_\nu-t|)\exp\Big(-\ga\sum_{j=1}^2|x(\nu)_j-t|\Big),
\end{multline}
with $c_4$ continuous on $[a_s/2,a_s)\times S_a$. To bound the remaining integral, we note that the integrand is invariant under the interchange $x(\nu)_1\leftrightarrow x(\nu)_2$, so that no generality is lost by assuming $x(\nu)_2\leq x(\nu)_1$. Then we can write the integral as a sum of three integrals
\be
I_n\equiv \int_{x(\nu)_n}^{x(\nu)_{n-1}}dt\,(1+|x_\nu-t|)\exp\Big(-\ga\sum_{j=1}^2|x(\nu)_j-t|\Big),\ \ \ n=1,2,3,
\ee
where $x(\nu)_0\equiv\infty$ and $x(\nu)_3\equiv-\infty$. For $I_1$, we have
\begin{multline}\label{I1b}
\exp(\ga|x(\nu)_1-x(\nu)_2|) I_1\\
=\exp(\ga(x(\nu)_1-x(\nu)_2))\int_{x(\nu)_1}^{\infty} dt\,(1+|x_\nu-t|)\exp(-\ga(2t-x(\nu)_1-x(\nu)_2))\\
=\int_0^\infty dt\,(1+|x_\nu-x(\nu)_1-t|)\exp(-2\ga t)\\
\leq\int_0^\infty dt\,(1+t+|x_\nu-x(\nu)_1|)\exp(-2\ga t)<C(1+|x_\nu-x(\nu)_1|),
\end{multline}
where we can take $C=(1+1/2\ga)/2\ga$. Similarly, we obtain
\be\label{I3b}
\exp(\ga|x(\nu)_1-x(\nu)_2|)I_3<C(1+|x_\nu-x(\nu)_2|).
\ee
In the case of $I_2$, we have
\begin{multline}\label{I2b}
\exp(\ga|x(\nu)_1-x(\nu)_2|) I_2=\int_{x(\nu)_2}^{x(\nu)_1}dt\,(1+|x_\nu-t|)\\
<\int_{x(\nu)_2}^{x(\nu)_1}dt\,\Big(1+\sum_{j=1}^2|x_\nu-x(\nu)_j|\Big)=(x(\nu)_1-x(\nu)_2)\Big(1+\sum_{j=1}^2|x_\nu-x(\nu)_j|\Big).
\end{multline}
Combining the bounds \eqref{hI3b} and \eqref{I1b}--\eqref{I2b}, we readily infer the majorization \eqref{rI3mbou}. 
\end{proof}

We conclude this section by deriving a uniform bound on $\rE_3(x,y)$, which is the counterpart of Prop.~\ref{Prop:rE2b}.


\begin{theorem}\label{Thm:ubound}
Letting $(\delta,b)\in (0,a_s]\times S_a$, we have
\be
|\rE_3(b;x,y)|<C(\de,b)\prod_{1\leq j<k\leq 3} \big(1+|\re(x_j-x_k)|\big)\cdot\exp\Big(-\alpha\sum_{j=1}^3y_jv_j\Big),
\ee
for all $(x,y)\in\C^3\times\R^3$ satisfying
\be\label{N3xyas}
v_j-v_k\in[-a_s+\de,0] ,\ \ \ 1\leq j<k\leq 3,\ \ d_3(y)>0,\ \  \ v=\im x,
\ee
where $C$ is a continuous function on $(0,a_s]\times S_a$.
\end{theorem}
\begin{proof}
We exploit once more  the representation for $\rE_3$ given by \eqref{rE3rep2}. Focusing first on the last sum, we begin by noting that the regularity of $u(b;x_k-x_j)$ for $-a_s<v_j-v_k<m(\re b)$ and the $u$-asymptotics \eqref{uas} entail 
 \be\label{ubo}
 |u(b;-z)|\le c(b,\im z),\ \ \ \ (b,\im z)\in S_a\times (-a_s,0],
 \ee
where $c(b,\im z)$ is continuous on $S_a\times (-a_s,0]$. Next, Prop.~\ref{Prop:rE2b} implies an estimate
\begin{eqnarray}\label{E2bo}
|\rE_2(x(\nu),(y_1-y_3,y_2-y_3))| & < &  C(\de,b)(1+ |\re (x(\nu)_1- x(\nu)_2)|)
\nonumber \\
 & \times & \exp\Big(-\alpha\sum_{k=1}^2(y_k-y_3)\im x(\nu)_k\Big).
\end{eqnarray}
Moreover, from \eqref{M3} and \eqref{crepb}, we obtain
\be
|M_3(b;y)\exp(iy_3(x_1+x_2+x_3))|\leq c(b)\exp\Big(-\alpha\sum_{j=1}^3y_jv_j\Big)\exp\Big(\alpha\sum_{k=1}^2(y_k-y_3)v_k\Big),
\ee
for all $(x,y)\in\C^3\times\R^3$, where $c$ is continuous on $S_a$. 

When we now take the product of the functions at hand and use
\begin{multline}
\exp\Big(\alpha\sum_{k=1}^2(y_k-y_3)v_k\Big) \exp\Big(-\alpha\sum_{k=1}^2(y_k-y_3)\im x(\nu)_k\Big)\\
=\exp\Big(\alpha\sum_{k=\nu }^2(y_k-y_3)(v_k-v_{k+1})\Big)\le  1, \ \ \ d_3(y)>0, \ v_k-v_{k+1}\le 0  ,\ k=1,2,
\end{multline}
then the desired bound for this contribution to~$\rE_3(b;x,y)$ easily follows.
As a consequence, it suffices to show that the 
 integrals appearing on the right-hand side of \eqref{rE3rep2} are bounded by  
\be\label{major}
C(\de,b)|\rho_3(b;y)|\exp\Big(-\alpha\sum_{k=1}^2(y_k-y_3)v_k\Big)\prod_{1\leq j<k\leq 3}(1+|\re(x_j-x_k)|),
\ee
for all $(x,y)\in\C^3\times\R^2$ satisfying \eqref{N3xyas}.  

Specializing the first equality in \eqref{rE3hom} to $\eta=-v_1$, it becomes clear that we may restrict attention to
\be
0\leq v_1\leq v_2\leq v_3\leq a_s-\de.
\ee
Requiring at first $x\in\R^3$, we begin by considering the integral of $\rI_3$ along the $z_k$-contours $C_b+ir$. Taking $z_k\to z_k+i(a- b/2+r)$ and making use of the identity \eqref{rE2hom}, we obtain again \eqref{rI3Expr}, but now with $r^\prime\to r$. Allowing next $v_j\neq 0$, we require
\be
\de^\prime\leq r-v_j\leq a_s-\de^\prime,\ \ \ \de^\prime\in(0,a_s/2],\ \ j=1,2,3,
\ee
in order to stay clear of the poles of the $G$-ratios for $z_k+ir-v_j=0,a_s$. By setting
\be\label{rde}
r=a_s-\de/2,\ \ \ \de^\prime=\de/2,
\ee
we can admit any $x\in\C^3$ satisfying the conditions in \eqref{N3xyas}. Using the bounds \eqref{rE2bex}, \eqref{crepb} and \eqref{Gratb}, we now infer 
\begin{multline}
\left|\int_{(C_b+ir)^2}dz\,\rI_3(x,y,z)\right|\leq c_2(\de,b)|\rho_3(y)|\exp(-\alpha r(y_1+y_2-2y_3))(1+y_1-y_2)\\
\times |C_3(b;x)|^{-1}\int_{\R^2}dz\,\frac{(z_1-z_2)\sinh(\ga(z_1-z_2))}{\prod_{j=1}^3\prod_{k=1}^2\cosh(\ga(\re x_j-z_k))},
\end{multline}
with $c_2$ continuous on $(0,a_s]\times S_a$.

Recalling the $c$-asymptotics~\eqref{cas} and the integral evaluation~\eqref{intEval}, we see that for the majorization~\eqref{major} to hold, it suffices to show that
\be
B\equiv \exp(-\alpha (r+v_1-v_2)(y_2-y_3)-\alpha r (y_1-y_3))(1+y_1-y_2)
\ee
is bounded. Now since $d_3(y)>0$ by assumption, we have
\be
B<\exp (-\alpha \de (y_2-y_3)/2-\alpha a_s(y_1-y_3)/2)(1+y_1-y_3)<C,
\ee
so this is indeed the case.

It remains to bound the integral of $\hat{\rI}_{3,\nu}$ along the $t$-contour $C_b+ir$. Taking $t\to t+i(a- b/2+r)$ and using once more the identity \eqref{rE2hom}, we obtain \eqref{rI3mExpr} with $r^\prime\to r$. It follows from~\eqref{rE2vb}  and the bounds \eqref{cas}, \eqref{Gratb} that we have
\begin{multline}
\left|\int_{C_b+ir}dt\,\hat{\rI}_{3,\nu}(x,y,t)\right|\leq c_3(\de,b)|\rho_3(y)|\exp\big(-\alpha[v_\nu(y_1-y_3)+r(y_2-y_3)]\big)\\
\times \exp(\ga|\re(x(\nu)_1-x(\nu)_2)|)\int_\R dt\big(1+|\re x_\nu-t|\big)\exp\big(-\ga|t-\re x(\nu)_1|-\ga|t-\re x(\nu)_2|\big),
\end{multline}
where $c_3$ is continuous on $(0,a_s]\times S_a$.

Finally, it follows from $v_\nu\geq v_1$ and $r>v_2$ that
$$
\exp\big(-\alpha[v_\nu(y_1-y_3)+r(y_2-y_3)]\big)<\exp\Big(-\alpha\sum_{k=1}^2(y_k-y_3)v_k\Big),
$$
and to bound the remaining integral, we can proceed as in the proof of Theorem \ref{Thm:rE3as}. Indeed, since the integrand is $x(\nu)_1\leftrightarrow x(\nu)_2$ invariant, we may assume $\re x(\nu)_2\leq \re x(\nu)_1$. Then writing $\R=(-\infty,\re x(\nu)_2)\cup[\re x(\nu)_2,\re x(\nu)_1)\cup[\re x(\nu)_1,\infty)$ and estimating the corresponding three integrals separately, we obtain the desired bound.
\end{proof}

\begin{appendix}

\section{The hyperbolic gamma function revisited}\label{AppA}
In the main text we need a few properties of the hyperbolic gamma function that were not mentioned in I Appendix A. They are collected in this appendix.

First, from Appendix A in \cite{R99} we recall that the hyperbolic gamma function can be written as a ratio of entire functions,
\be\label{GE}
G(a_+,a_-;z)=E(a_+,a_-;z)/E(a_+,a_-;-z),
\ee
with the zeros of $E(a_+,a_-;z)$ located at
\be\label{Ezs}
z=ia+ip_{kl},\ \ \ k,l\in\N,
\ee
where
\be\label{pkl}
p_{kl}\equiv ka_++la_-.
\ee
The order of these zeros equals the number of distinct pairs $(m,n)\in\N^2$ such that $p_{mn}=p_{kl}$. In particular, for $a_+/a_-\notin\Q$ all zeros are simple.

The function $E(z)\equiv E(a_+,a_-;z)$ from~\cite{R99} we employ in this paper  is a cousin of Barnes' double gamma function. It has no zeros for $z$ in the half plane
\be
\Lambda\equiv \{z\in\C\mid \im z<a\},
\ee
so it can be written as
\be
E(z)=\exp(e(z)),\ \ \ z\in \Lambda,
\ee
with $e(z)$ holomorphic in $\Lambda$. Explicitly, $e(z)$ has the integral representation
\be
e(a_+,a_-;z)=\frac{1}{4}\int_0^\infty \frac{dy}{y}\left(\frac{1-e^{-2iyz}}{\sh a_+y~\sh a_-y}-\frac{2iz}{a_+a_-y}-\frac{z^2}{a_+a_-}(e^{-2a_+y}+e^{-2a_-y})\right).
\ee

A distinguishing feature of this $E$-function  is that it satisfies the two A$\De$Es
\be\label{EADE}
\frac{E(z+ia_{-\de}/2)}{E(z-ia_{-\de}/2)}=\frac{\sqrt{2\pi}}{\Gamma(iz/a_{\de}+1/2)}\exp(izK_\de),\ \ \ \de=+,-,
\ee
where
\be
K_\de\equiv \frac{1}{2a_{\de}}\ln\left(\frac{a_{-\de}}{a_{\de}}\right).
\ee
We need one of these A$\De$Es in Subsections \ref{Sec22} and \ref{Sec32}.
 
Finally, we have occasion to make use of the Fourier transform formula from Proposition C.1 in \cite{R11}. Specifically, let $\mu,\nu\in\C$ be such that
\be
-a<\im\mu<\im\nu<a,
\ee
and assume that $y\in\C$ satisfies
\be
|\im y|<\im(\nu-\mu)/2.
\ee
Then the pertinent formula can be written
\begin{multline}\label{Fform}
\left(\frac{\alpha}{2\pi}\right)^{1/2}\int_\R dz\exp(i\alpha pz)\frac{G(z-\nu)}{G(z-\mu)}\\
=\exp(i\alpha p(\mu+\nu)/2)G(ia+\mu-\nu)\prod_{\de=+,-}G(\de p-ia+(\nu-\mu)/2).
\end{multline}

\end{appendix}

\bibliographystyle{amsalpha}

\end{document}